%% file: arxiv.tex
\documentclass[11pt, fleqn]{article}
\usepackage[english]{babel}

\usepackage[lmargin=1.1in,rmargin=1.1in,bottom=1.3in,top=1.3in,
twoside=False]{geometry}

\usepackage{mathrsfs}
\usepackage{relsize,xspace}
 \usepackage{xcolor}
 \usepackage{mathtools}
 \usepackage{todonotes}
 \usepackage{comment}
\usepackage{microtype}
\usepackage{amsmath}
\usepackage{amssymb}
\usepackage{amsfonts}
\usepackage{stmaryrd}
\usepackage{bm}
\usepackage{tikz}
\usepackage{refcount}
\usepackage{wrapfig}

\usepackage{marginnote}

\definecolor{sz}{rgb}{0.1,0.2,0.6}
\definecolor{blue}{rgb}{0.1,0.2,0.5}
\definecolor{brown}{rgb}{0.6,0.6,0.2}

\usepackage[amsmath,thmmarks,hyperref]{ntheorem}
\usepackage{cleveref}

\crefformat{page}{#2page~#1#3}%
\Crefformat{page}{#2Page~#1#3}%
\crefformat{equation}{#2(#1)#3}%
\Crefformat{equation}{#2(#1)#3}%
\crefformat{figure}{#2Figure~#1#3}%
\Crefformat{figure}{#2Figure~#1#3}%
\crefformat{section}{#2Section~#1#3}
\Crefformat{section}{#2Section~#1#3}
\crefformat{appendix}{#2Appendix~#1#3}
\Crefformat{appendix}{#2Appendix~#1#3}
\crefformat{chapter}{#2Chapter~#1#3}
\Crefformat{chapter}{#2Chapter~#1#3}
\crefformat{chapter*}{#2Chapter~#1#3}
\Crefformat{chapter*}{#2Chapter~#1#3}
\crefformat{part}{#2Part~#1#3}
\Crefformat{part}{#2Part~#1#3}
\crefformat{enumi}{#2(#1)#3}
\Crefformat{enumi}{#2(#1)#3}

\usepackage{enumerate}

\usepackage{latexsym}

\theoremnumbering{arabic}
\theoremstyle{plain}
\theoremsymbol{}
\theorembodyfont{\itshape}
\theoremheaderfont{\normalfont\bfseries}
\theoremseparator{.}

\newtheorem{theorem}{Theorem}
\crefformat{theorem}{#2Theorem~#1#3}
\Crefformat{theorem}{#2Theorem~#1#3}
\newcommand{\newtheoremwithcrefformat}[2]{%
  \newtheorem{#1}[lemma]{#2}%
  \crefformat{#1}{##2\MakeUppercase#1~##1##3}%
  \Crefformat{#1}{##2\MakeUppercase#1~##1##3}%
}
\newcommand{\newseptheoremwithcrefformat}[2]{%
  \newtheorem{#1}{#2}%
  \crefformat{#1}{##2\MakeUppercase#1~##1##3}%
  \Crefformat{#1}{##2\MakeUppercase#1~##1##3}%
}

\newseptheoremwithcrefformat{lemma}{Lemma}
\newtheoremwithcrefformat{proposition}{Proposition}
\newtheoremwithcrefformat{observation}{Observation}
\newtheoremwithcrefformat{conjecture}{Conjecture}
\newtheoremwithcrefformat{corollary}{Corollary}
\newseptheoremwithcrefformat{claim}{Claim}
\newseptheoremwithcrefformat{goal}{Goal}
\theorembodyfont{\upshape}
\newtheoremwithcrefformat{example}{Example}
\newtheoremwithcrefformat{remark}{Remark}
\newseptheoremwithcrefformat{definition}{Definition}

\theoremstyle{nonumberplain}
\theoremheaderfont{\scshape}
\theorembodyfont{\normalfont}
\theoremsymbol{\ensuremath{\square}}
\newtheorem{aloneproof}{Proof}

\newtheorem{proof}{Proof}

\theoremsymbol{\ensuremath{\lrcorner}}
\newtheorem{clproof}{Proof}

\newcommand{\AND}{\mathrm{AND}}
\newcommand{\OR}{\mathrm{OR}}
\newcommand{\NOT}{\mathrm{NOT}}
\newcommand{\TRUE}{\mathrm{TRUE}}
\newcommand{\FALSE}{\mathrm{FALSE}}
\newcommand{\inpv}{\mathrm{in}}
\newcommand{\outv}{\mathrm{out}}

\newcommand{\backconn}{\mathrm{bconn}}
\newcommand{\dom}{\mathrm{dom}}
\newcommand{\wcol}{\mathrm{wcol}}

\newcommand{\adm}{\mathrm{adm}}

\newcommand{\WReach}{\mathrm{WReach}}

\newcommand{\CCC}{\mathscr{C}}

\newcommand{\dlogtime}{$\mathsf{dlogtime}$}

\newcommand{\linFPT}{\mathsf{linFPT}}
\newcommand{\FPT}{\mathsf{FPT}}
\newcommand{\XP}{\mathsf{XP}}
\newcommand{\AC}{$\mathsf{AC}$}

\newcommand{\NC}{$\mathsf{NC}$}

\newcommand{\paraAC}[1]{\mathrm{para}\textrm{-}\mathsf{AC}^{#1}}
\newcommand{\paraNC}[1]{\mathrm{para}\textrm{-}\mathsf{NC}^{#1}}

\newcommand{\Oh}{\mathcal{O}}
\newcommand{\Cc}{\mathscr{C}}
\newcommand{\Pp}{\mathcal{P}}
\newcommand{\Qq}{\mathcal{Q}}
\newcommand{\Rr}{\mathcal{R}}
\newcommand{\topnabla}{\widetilde{\nabla}}

\newcommand{\I}{\mathbb{I}}
\newcommand{\N}{\mathbb{N}}

\renewcommand{\phi}{\varphi}
\renewcommand{\epsilon}{\varepsilon}

\newcommand{\FO}{\mathrm{FO}}
\newcommand{\MSO}{\mathrm{MSO}}

\newcommand{\ar}{\mathrm{ar}}

\newcommand{\from}{\colon}
\newcommand{\eps}{\varepsilon}

\renewcommand{\leq}{\leqslant}
\renewcommand{\geq}{\geqslant}

\newcommand{\As}{\mathbb{A}}
\newcommand{\Bs}{\mathbb{B}}
\newcommand{\Cs}{\mathbb{C}}
\newcommand{\Ss}{\mathbb{S}}
\newcommand{\Ts}{\mathbb{T}}
\newcommand{\Us}{\mathbb{U}}

\newcommand{\basic}{\mathsf{type}}
\newcommand{\Basic}{\mathbf{Types}}

\newcommand{\wLambda}{\widehat{\Lambda}}

\newcommand{\wphi}{\widehat{\varphi}}
\newcommand{\lcd}{\mathsf{lcd}}
\newcommand{\parent}{\mathsf{parent}}
\newcommand{\Classes}{\mathcal{U}}
\newcommand{\clr}{\mathsf{class}}

\hypersetup{ocgcolorlinks=true,colorlinks=true,linkcolor={blue}, citecolor={brown}}

\title{Parameterized circuit complexity of model checking\\ first-order logic on sparse structures
\thanks{
The work of M.\ Pilipczuk and S.\ Siebertz is supported by the National Science Centre of 
Poland via POLONEZ grant agreement UMO-2015/19/P/ST6/03998, 
which has received funding from the European Union's Horizon 2020 research and 
innovation programme (Marie Sk\l odowska-Curie grant agreement No.\ 665778).
The work of Sz.~Toru{\'n}czyk is supported by the National Science Centre of Poland grant 2016/21/D/ST6/01485.
}}

\author{
Micha\l~Pilipczuk \qquad
\qquad Sebastian Siebertz
\qquad Szymon Toru{\'n}czyk\\[0.3cm]
Institute of Informatics, University of Warsaw, Poland\\[0.1cm]
\texttt{\{michal.pilipczuk,siebertz,szymtor\}@mimuw.edu.pl}}

\begin{document}

\maketitle
\input{abstract}

\begin{picture}(0,0) \put(395,-258)
{\hbox{\includegraphics[scale=0.25]{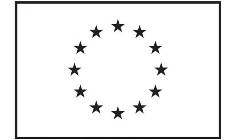}}} \end{picture} 
\vspace{-0.8cm}


\input{introduction}
\input{preliminaries}
\input{orderings}

\input{low-td}

\input{applications}

\input{conclusions}

\paragraph*{Acknowledgements.}
The authors thank 
Thomas Zeume for discussions on dynamic $\FO$ in the context of sparse graphs, which inspired this work.

\bibliographystyle{abbrv}

\pagebreak
\bibliography{ref} 

\newpage

\appendix

\input{app-degeneracy}

\input{app-adm}

\end{document}

%% file: abstract.tex
\begin{abstract}
\noindent 
We prove that for every class $\Cc$ of graphs with effectively bounded expansion,
given a first-order sentence $\varphi$ and an $n$-element structure $\As$ whose Gaifman graph belongs to~$\Cc$, the question whether
$\varphi$ holds in $\As$ can be decided by a family of \AC-circuits of size $f(\varphi)\cdot n^c$ and depth $f(\varphi)+c\log n$, where $f$ is a computable function and $c$ is a universal constant.
This places the model-checking problem for classes of bounded expansion in the parameterized circuit complexity class $\paraAC{1}$.
On the route to our result we prove that the basic decomposition toolbox for classes of bounded expansion, including orderings with bounded weak coloring numbers and low treedepth decompositions,
can be computed in $\paraAC{1}$.
\end{abstract}

%% file: introduction.tex
\section{Introduction}\label{sec:intro}

\paragraph*{Model-checking on sparse structures.} 
We study the model-checking problem for first-order logic ($\FO$): given a relational structure $\As$ and a first-order sentence $\varphi$ over the vocabulary of~$\As$, decide whether $\varphi$ holds in $\As$.
A naive algorithm for this problem recursively browses through all evaluations of each quantified variable and runs in time $n^{\Oh(|\varphi|)}$, where $n$ is the size of the universe of $\As$.
Thus, the running time is polynomial for every fixed $\varphi$, but the degree of the polynomial depends on $\varphi$. In the language of parameterized complexity,
this means that the model-checking problem for first-order logic on arbitrary structures is in the complexity class $\XP$ when parameterized by the input formula $\varphi$.
This is traditionally put in contrast to the class $\FPT$ (for {\em{fixed-parameter tractable}})
where we require the existence of an algorithm with running time $f(\varphi)\cdot n^c$ for a computable function $f$ and universal constant
$c$; thus, the degree of the polynomial factor has to be independent of the parameter. See~\cite{platypus,DowneyF13,FlumG06} for an introduction to parameterized complexity.

\pagebreak
In general structures we do not hope for an $\FPT$ algorithm for model-checking $\FO$, because the problem is complete for the class $\mathsf{AW}[\star]$ (cf.~\cite{FlumG06}). Already the problem of deciding the  existence of a clique of size $k$ in a given graph of size $n$, which is easily expressible by a first-order formula with~$k$ existential quantifiers, is $\mathsf{W}[1]$ hard in general, so believed not to be
$\FPT$ when parametrized by $k$, i.e., solvable by an algorithm with running time $f(k)\cdot n^c$ for some computable~$f$ and fixed $c$.
However, it was realized that on {\em{sparse}} structures, i.e.\  those whose Gaifman graph is sparse,
efficient parameterized algorithms for model-checking $\FO$ exist. Starting with the result of Seese~\cite{see96}, who gave an
FPT algorithm on structures with universally bounded degree, a long line of research focused on showing fixed-parameterized tractability 
of the problem for more and more general classes of sparse structures: of bounded local treewidth~\cite{frick2001deciding} (this includes planar and bounded-genus structures), 
excluding a fixed minor~\cite{FlumG01}, and locally excluding a minor~\cite{dawar2007locally}.

This line of research naturally converged to studying abstract notions of sparsity: classes of {\em{bounded expansion}} and {\em{nowhere dense}} classes.
These two concepts form foundations of a deep and rapidly developing theory of sparse graph classes, first introduced and pursued
by Ne\v{s}et\v{r}il and Ossona de Mendez~\cite{nevsetvril2008grad,nevsetvril2010first,nevsetvril2011nowhere,NesetrilM08a}, which by now has found multiple applications in combinatorics, algorithm design, and logic.
We refer the reader to the book of Ne\v{s}et\v{r}il and Ossona de Mendez~\cite{sparsity} for a comprehensive overview of the field as of 2012,
and to the lecture notes of the first two authors for a compact and updated exposition of the basic toolbox~\cite{notes}.

Formally, a graph~$H$ is a {\em{depth-$r$ minor}} of a graph~$G$ if~$H$ can be obtained from a subgraph of~$G$ by contracting mutually disjoint connected subgraphs of radius at most $r$.
A class of graphs~$\Cc$ has {\em{bounded expansion}} if there is a function $f\colon \N\to \N$ such that for every $r\in \N$,
in every depth-$r$ minor of a graph from $\Cc$ the ratio between the number of edges and the number
of vertices is bounded by $f(r)$. More generally, $\Cc$ is {\em{nowhere dense}} if there is a function $t\colon \N\to \N$ such that no graph from $\Cc$ admits the clique $K_{t(r)}$ as a depth-$r$ minor. 
Every class of bounded expansion is nowhere dense, but the converse does not necessarily hold~\cite{nevsetvril2011nowhere}.
Class~$\Cc$ has {\em{effectively bounded expansion}}, respectively is {\em{effectively nowhere dense}}, if the respective function~$f$ or $t$ as above is computable.
These definitions are naturally generalized to classes of relational structures by considering the Gaifman graph of a structure.

Many classes of sparse graphs studied in the literature have (effectively) bounded expansion. These include: planar graphs, graphs of bounded maximum degree, graphs of bounded treewidth, and more generally, graphs excluding a fixed (topological) minor.
A notable negative example is that classes with bounded {\em{degeneracy}}, equivalently with bounded {\em{arboricity}},
do not necessarily have bounded expansion, as there we have only a finite bound on the edge density in subgraphs (aka depth-$0$ minors).

By the result of Dvo\v{r}\'ak, Kr\'al', and Thomas~\cite{dvovrak2013testing}, the $\FO$ model-checking
problem on any class~$\Cc$ of effectively bounded expansion admits a linear $\FPT$ algorithm, 
i.e.\  with running time $f(\varphi)\cdot n$
for computable~$f$. Grohe, Kreutzer, and the second author~\cite{grohe2014deciding} lifted this result to any effectively nowhere dense class $\Cc$; here, the dependence on the structure size $n$ is {\em{almost linear}}, i.e.\ 
of the form $n^{1+\eps}$ for any $\eps>0$. As observed by Dvo\v{r}\'ak et al.~\cite{dvovrak2013testing}, the result of Grohe et al.~\cite{grohe2014deciding} is the final answer as long as subgraph-closed classes
are concerned: on any subgraph-closed class $\Cc$ that is not nowhere dense, the $\FO$ model-checking problem is already $\mathsf{AW}[\star]$-complete. Conceptually, this means that the notion of nowhere denseness
exactly characterizes classes of inputs where sparsity-based arguments can lead to efficient parameterized algorithms for deciding first-order definable properties.

\paragraph*{Parameterized circuit complexity.} 
In this paper we take a different angle on the parameterized complexity of model-checking $\FO$, namely that of {\em{circuit complexity}}.
A fundamental fact from descriptive complexity is that $\FO$ is essentially equivalent to $\mathsf{AC}^0$ (c.f.~\cite{0095988} for a precise statement).
In particular, every fixed first-order expressible property can be checked by a family of \AC-circuits of polynomial size and
constant depth. More precisely, provided the property is expressed by a first-order sentence $\varphi$, the circuit for $n$-element inputs has
size $n^{\Oh(|\varphi|)}$ and depth $\Oh(|\varphi|)$. Viewing this via the standard interpretation of circuits as an abstraction for parallel algorithms, this is a highly parallelizable algorithm performing total
(sequential) XP work. Obviously, in general we cannot expect the problem to be solvable by circuits of $\FPT$ size (i.e.,  of size $f(\varphi)\cdot n^c$ for computable~$f$ and constant $c$), 
as evaluating such circuits would yield a sequential $\FPT$ algorithm, implying $\FPT=\mathsf{AW}[\star]$. However, the question for known classes for which $\FO$ model-checking is $\FPT$ persists: how, and in what sense,
can we solve $\FO$ model-checking on these classes using circuits of $\FPT$ size and low depth? Viewing circuits again as a model for parallelization, this would correspond to a well-parallelizable $\FPT$ algorithm.

Curiously, even though the complexity-theoretical foundations of parameterized complexity are expressed using circuit complexity, the question of what are the appropriate analogues of standard circuit complexity classes in
parameterized complexity was not systematically studied up to very recently, when Elberfeld et al.~\cite{ElberfeldST15} and Bannach et al.~\cite{BannachST15} introduced an appropriate definitional layer and gave several foundational results.
Slightly informally, a parameterized problem is in the class $\paraAC{i}$ (where $i>0$) if it can be solved by an (appropriately uniform)
family of \AC-circuits $(C_{n,k})_{n,k\in \N}$, where the circuit $C_{n,k}$ solves the problem on inputs of size $n$ and parameter value $k$,
such that each $C_{n,k}$ has size $f(k)\cdot n^c$ and depth $f(k)+c\cdot \log^i n$, for a computable function $f$ and universal constant $c$.
The classes $\paraNC{i}$ are defined similarly using \NC-circuits. The class $\paraAC{0}$ is defined slightly differently: we require the depth to be bounded by a universal constant, independent of the parameter.
By allowing the depth to be bounded by a function of the parameter we obtain the larger class $\paraAC{0\uparrow}$. We give the formal definitions and fix our notation in Section~\ref{sec:prelims}.

In~\cite{BannachST15} Bannach et al.\ showed how the technique of {\em{color coding}} can be implemented in $\paraAC{0}$, leading to a first batch of results for several parameterized problems.
Later, Bannach and Tantau~\cite{BannachT16} showed that model-checking mona\-dic second-order logic ($\MSO$) can be done in $\paraAC{0\uparrow}$ on structures of bounded treedepth
and in $\paraNC{2+\eps}$ for any \mbox{$\eps>0$} on structures of bounded treewidth; here, the parameter is both the formula and the treedepth, respectively the treewidth of the input structure.
This is a parameterized circuit complexity analogue of the classic  theorem of Courcelle stating that model-checking $\MSO$ is fixed-parameter tractable when parameterized by the formula and the tree\-width of the input
structure. Recent advances show descriptive relations between parameterized circuit complexity and fragments of $\FO$ with a bounded number of variables~\cite{ChenFH17}, 
and applications to kernelization~\cite{BannachT18}.

In this light, it is natural to ask about the parameterized circuit complexity of model checking~$\FO$ on sparse structures. Investigating this question is precisely the goal of this work.
%
Our main result is encompassed by the following theorem.

\begin{theorem}\label{thm:main}
Suppose $\Cc$ is a graph class with effectively bounded expansion and let $\Sigma$ be a relational vocabulary of arity~$2$. 
Then the following problem parameterized by $\varphi\in \FO[\Sigma]$ is in $\paraAC{1}$: given a $\Sigma$-structure $\As$ whose Gaifman graph belongs to $\Cc$, determine whether $\As\models \varphi$.
\end{theorem}

\pagebreak Unraveling the definitions, Theorem~\ref{thm:main} states that model-checking $\FO$ on $\Sigma$-structures with Gaifman graphs belonging to $\Cc$ can be done using a family of \AC{} circuits $(C_{n,\varphi})_{n\in \N,\varphi\in \FO[\Sigma]}$,
where $C_{n,\varphi}$ verifies the satisfaction of $\varphi$ on structures with $n$ elements, and each $C_{n,\varphi}$ has size $f(\varphi)\cdot n^c$ and depth $f(\varphi)+c\log n$, for a computable function $f$ and universal
constant~$c$. Viewing circuits as an abstraction for parallel algorithms, this means that the problem can be solved in parallel time $f(\varphi)+c\log n$ and performing total work $f(\varphi)\cdot n^c$.
Hence Theorem~\ref{thm:main} can be regarded as a parallelized variant of the result of Dvo\v{r}\'ak et al.~\cite{dvovrak2013testing}.

The assumption in Theorem~\ref{thm:main} that $\Sigma$ has arity $2$ allows us to abstract away the question of how the input is represented, as we simply assume that each
relation in $\As$ is encoded on input as a one- or two-dimensional boolean table. 
In the presence of higher-arity relations, the choice of an encoding could influence the statements about bounds on circuit sizes in a technical way. We prefer to avoid
 these issues and simply assume that there are no higher-arity relations.

\paragraph*{Our techniques.} We prove Theorem~\ref{thm:main} by analyzing the existing approach to proving fixed-parameter tractability in the sequential case.
Essentially, the idea is to first compute a suitable decomposition of the input structure, and then leverage this decomposition to give a quantifier elimination procedure for first-order logic.
A suitable decomposition has the form of a {\em{low treedepth coloring}} that uses a bounded number of colors; 
it is known that such colorings exist for graphs from classes of bounded expansion~\cite{nevsetvril2008grad}. 
%
Efficient algorithms for computing low treedepth colorings provided by Ne\v{s}et\v{r}il and Ossona de Mendez~\cite{nevsetvril2008grad} allowed
Dvo\v{r}\'ak et al.~\cite{dvovrak2013testing} to give an efficient quantifier elimination procedure that reduces every first-order formula to a quantifier-free formula
at the cost of extending the structure by adding new unary relations and unary functions, which however do not change the Gaifman graph. From this, an algorithm for model-checking follows.
Later, Grohe and Kreutzer~\cite{grokre11} gave a new presentation of the quantifier elimination procedure, which is conceptually quite different from the original argument of~\cite{dvovrak2013testing}.
In particular it reduces every formula to an existential formula instead of quantifier-free one, but the extension of the structure does not use  function symbols.

Our work toward the proof of Theorem~\ref{thm:main} is divided into two parts.
First, we prove that a low treedepth coloring of the graph can be computed in $\paraAC{1}$.
Second, using this result we revisit the quantifier elimination procedure and show that it can be implemented in $\paraAC{1}$.

%
%


For computing a low treedepth coloring, on high level we follow the standard approach: first find a vertex ordering of the given graph with bounded {\em{weak coloring number}}, 
and then apply a coloring procedure on this vertex ordering to get a low treedepth coloring.
Classic implementations of both these steps are sequential, however we show that both of them can be performed in $\paraAC{1}$.
For the first step, the main idea is to construct the ordering by extracting the vertices not one by one, as a sequential algorithm would do, but in much larger chunks.
Namely, we perform $\log n$ rounds where in each round at least half of the remaining vertices are extracted and ordered, which directly translates to a construction of a $\paraAC{1}$ circuit family.
This comes at a price in the quality of the obtained ordering; in other words, we trade the approximation factor of the algorithm for its parallelization.
For the second step, we follow the classic divide-and-conquer approach to parallel coloring graphs of bounded maximum degree. Then we extend this to graphs of bounded degeneracy by
applying essentially the same technique of dividing the graph into $\log n$ parts, each inducing a graph of bounded maximum degree.

\pagebreak
We remark that the approach explained above is heavily inspired by the existing body of work on {\em{distributed algorithms}}
for sparse graphs. 
%
We relied on ideas from Barenboim and Elkin~\cite{BarenboimE10} who, among other results,
gave an $\Oh(\log n)$-time distributed algorithm that, given a graph of degeneracy $d$, finds its proper coloring with $\Oh(d^2)$ colors\footnote{Barenboim and Elkin use the parameter {\em{arboricity}} which differs
from degeneracy by multiplicative factor at most~$2$.}.
In particular, Barenboim and Elkin showed how to compute an approximate degeneracy ordering of such a graph in $\Oh(\log n)$ communication rounds by extracting half of the remaining vertices in each round;
our algorithm for this step is a circuit implementation of this procedure, lifted to the weak coloring number instead of degeneracy.
Using the results of~\cite{BarenboimE10} and the approach via fraternal augmentations, 
Ne\v{s}et\v{r}il and Ossona de Mendez~\cite{NesetrilM16} gave a logarithmic-time distributed algorithm that, given a graph from a fixed class of bounded expansion,
computes its treedepth-$p$ coloring using a constant number of colors, for any constant~$p$.


For the quantifier elimination procedure, we essentially revisit the existing approach and show that it can be implemented in $\paraAC{1}$.
This requires technical attention in several places, but conceptually there is no new ingredient.
Our argument is roughly based on the exposition of Grohe and Kreutzer~\cite{grokre11}. However, we obtain a stronger final form, similar to that of Dvo\v{r}\'ak et al.~\cite{dvovrak2013testing},
replace the usage of $\FO$ types with an explicit combinatorial argument in the spirit of marking witnesses as in~\cite{dvovrak2013testing}, and streamline the presentation.

\paragraph*{Additional results.}
We believe that our approach to the proof of Theorem~\ref{thm:main} has an additional benefit in that we implement most of the basic algorithmic toolbox for classes of bounded expansion in $\paraAC{1}$.
This can be re-used for problems other than model-checking first-order logic.

For instance, consider the problem of computing the smallest {\em{distance-$r$ dominating set}} in a given graph $G$, which is a subset of vertices $D$ such that every vertex of $G$ is at distance
at most $r$ from some vertex of $D$. This problem often serves as a benchmark for sparsity methods, and it was considered in the theory of sparse graphs from the points of view of
parameterized algorithms~\cite{DawarK09}, approximation~\cite{AmiriMRS17,Dvorak13}, and kernelization~\cite{drange2016kernelization,eickmeyer2016neighborhood,siebertz2016polynomial}.
In particular, Amiri et al.~\cite{AmiriMRS17} have recently used the results of Ne\v{s}et\v{r}il and Ossona de Mendez~\cite{NesetrilM16} to give a distributed logarithmic-time 
constant-factor approximation algorithm for the distance-$r$ dominating set problem on any class of bounded expansion. 
By combining their ideas with our constructions of orderings with bounded weak coloring numbers, we immediately obtain
the following approximation result.

\begin{theorem}\label{thm:domset-main}
Suppose $\Cc$ is a graph class with effectively bounded expansion.
Then there exists a computable function $\alpha\colon \N\to \N$ such that 
the following problem parameterized by $r\in \N$ is in $\paraAC{1}$: given a graph $G\in \Cc$, compute a distance-$r$ dominating set in $G$ of size at most $\alpha(r)\cdot \dom_r(G)$.
Here, $\dom_r(G)$ denotes the size of a smallest distance-$r$ dominating set in $G$.
\end{theorem}

\paragraph*{Organization.} In Section~\ref{sec:prelims} we establish notation and recall known results.
In Section~\ref{sec:orderings} we show how to compute vertex orderings with low weak coloring numbers,
while in Section~\ref{sec:low-td} we show how to construct low treedepth colorings.
Section~\ref{sec:mc} contains the proof of the main result, Theorem~\ref{thm:main}.
We conclude in Section~\ref{sec:conclusions} with some final remarks and open problems.

%% file: preliminaries.tex
\section{Preliminaries}\label{sec:prelims}

%
All graphs considered in this paper are simple, i.e. do not contain self-loops or multiple edges connecting the same pair of vertices.
We use standard graph notation, see e.g.~\cite{platypus}. 

\paragraph*{Classes of bounded expansion.}
As mentioned in the introduction, a graph $H$ is a {\em{depth-$r$ minor}} of a graph $G$ if $H$ can be obtained from a subgraph of $G$ by contracting mutually disjoint connected subgraphs of radius at most $r$.
Graph $H$ is a {\em{depth-$r$ topological minor}} of~$G$ if there is a subgraph of~$G$ that is an {\em{$\leq 2r$-subdivision}} of $H$, that is, can be obtained from~$H$ by replacing each edge by a path
of length at most $2r+1$. It is easy to see that if $H$ is a depth-$r$ topological minor of $G$, then $H$ is also a depth-$r$ minor of $G$.

A graph $H$ is a {\em{depth-$r$ minor}} of $G$ if one can find a {\em{depth-$r$ minor model}} $(I_u)_{u\in V(H)}$  of $H$ in~$G$,
where $I_u$ for $u\in V(H)$ are pairwise vertex-disjoint connected subgraphs of $G$ of radius at most $r$ such that for each $uv\in E(H)$ there is an edge between a vertex of $I_u$ and a vertex of~$I_v$.
In other words, in the standard definition of a minor model we restrict the {\em{branch sets}} $I_u$ to have radius at most $r$.
Similarly, $H$ is a {\em{depth-$r$ topological minor}} of $G$ if there exists a {\em{depth-$r$ topological minor model}} $\mu$ of $H$ in $G$:
such $\mu$ is a mapping that sends vertices of $H$ to pairwise different vertices of $G$ and edges of $H$ to paths of length at most $2r+1$ in $G$ such that $\mu(uv)$ has endpoints $\mu(u)$ and $\mu(v)$ for each
$uv\in E(H)$, and paths in $\mu(E(H))$ pairwise may share only the endpoints.

The {\em{edge density}} of a graph $G$ is the ratio between the number of edges and the number of vertices in $G$, i.e., $\frac{|E(G)|}{|V(G)|}$.
For a graph $G$, by $\nabla_r(G)$ we denote the maximum over all depth-$r$ minors $H$ of $G$ of the edge density of $H$. 
Similarly, $\topnabla_r(G)$ denotes the maximum edge density in a depth-$r$ topological minor
of $G$. As a depth-$r$ topological minor is also a depth-$r$ minor, we have $\topnabla_r(G)\leq \nabla_r(G)$. However, it is known that $\nabla_r(G)$ is also bounded from above
by a computable function of $r$ and $\topnabla_r(G)$~\cite{dvorak2007asymptotical}. 
For a graph class $\Cc$, we denote $\nabla_r(\Cc)=\sup_{G\in \Cc} \nabla_r(G)$.
 and $\topnabla_r(\Cc)=\sup_{G\in \Cc} \topnabla_r(G)$.
A graph class~$\Cc$ has {\em{bounded expansion}} if $\nabla_r(\Cc)$ is finite for all $r\in \N$, equivalently if $\topnabla_r(\Cc)$ is finite for all $r\in \N$.
It has {\em{effectively bounded expansion}} if there is a computable (from~$r$) upper bound on $\nabla_r(\Cc)$, equivalently on $\topnabla_r(\Cc)$.

\paragraph*{Parameterized circuit complexity.} 
We explain the definitional layer of parameterized circuit complexity basing our notation on Bannach et al.~\cite{BannachST15}, 
though prior foundational work on parameterized circuit complexity was done by Elberfeld et al.~\cite{ElberfeldST15}.
An {\em{\AC-circuit}} $C$ is a directed acyclic graph with node set consisting of input, conjunction ($\AND$), disjunction ($\OR$), and negation ($\NOT$) gates.
There are no restrictions on the fan-in or fan-out of the gates.
One or more sources of $C$ are designated as the output gates, and both input and output gates of $C$ are ordered. 
If $(u_1,\ldots,u_n)$ and $(v_1,\ldots,v_m)$ are the input and output gates of $C$, respectively, then~$C$ evaluates a function from $\{0,1\}^n$ to $\{0,1\}^m$ defined as follows:
given input $x=(x_1,\ldots,x_n)$, set the value of each input gate $u_i$ to $x_i$, evaluate the gates of the circuit in a bottom-up manner naturally, and define the output $y$ to be
the sequence of values computed in gates $(v_1,\ldots,v_m)$. 
The {\em{depth}} of a circuit is the length of a longest path from an input gate to an output gate.
The {\em{size}} of a circuit is the number of its gates.

\pagebreak
A {\em{parameterized transformation}} is a function $F\colon \{0,1\}^\star\to \{0,1\}^\star$ together with a polynomial-time computable function $\kappa\colon \{0,1\}^\star\to \I$, called {\em{parameterization}}. 
Here $\I$ is some indexing set for parameters and we assume that its elements
can be encoded as binary strings. 
Typically in parameterized complexity we have $\I=\N$ --- the parameter is just an integer --- but it will be convenient to assume larger generality, as some our circuit
families will be indexed by first-order sentences on graphs or tuples of integers. 
A {\em{parameterized problem}} is just a parameterized transformation with the output always belonging to $\{0,1\}$, for false and true, respectively.

A parameterized transformation is in the class $\FPT$ if there is an algorithm that computes it in time $f(k)\cdot n^c$ on inputs of size $n$ and parameter value $k$, where $f$ is a computable function and $c$ is
a universal constant. It is in class $\linFPT$ if moreover $c=1$ and $\kappa(\cdot)$ is linear-time computable.

For an indexing set $\I$, we may consider a family $(C_{n,k})_{n\in \N,k\in \I}$ of \AC-circuits, where each $C_{n,k}$ has exactly $n$ inputs.
We say that such a family is {\em{\dlogtime-uniform}} if there exists an algorithm that given $n\in \N$, $k\in \I$, and $i\in \N$, all encoded in binary, computes the $i$-th bit of the encoding of $C_{n,k}$ 
in time $f(k)+\Oh(\log i+\log n)$, for some computable function $f$. All circuit families in this paper are \dlogtime-uniform; this will always follow from the construction in a straightforward manner, so
we refrain from providing technical details in order not to obfuscate the main ideas.

For $i>0$, we say that a parameterized transformation $(F,\kappa)$ is in $\paraAC{i}$ if there exists a \dlogtime-uniform family of circuits $(C_{n,k})_{n\in \N,k\in \I}$ 
such that 
\begin{itemize}
\item $C_{n,k}$ has size $f(k)\cdot n^{\Oh(1)}$ and depth $f(k)+\Oh(\log^i n)$ for some computable function $f\colon \I\to \N$;
\item for each $x\in \{0,1\}^\star$, the output of $C_{|x|,\kappa(x)}$ applied to $x$ is $F(x)$.
\end{itemize}

Note that the above definition implicitly assumes that for all $x,x'$ with $|x|=|x'|$ and $\kappa(x)=\kappa(x')$, the outputs $F(x)$ and $F(x')$ have the same length, as this length must be equal to the number
of outputs of $C_{|x|,\kappa(x)}$. 
Bannach et al.~\cite{BannachST15} also define a larger class $\paraAC{i\uparrow}$ by relaxing the restriction on the depth from $f(k)+\Oh(\log^i n)$ to $f(k)\cdot \log^i n$, for a computable function $f$.
It is easy to see that $\paraAC{i}\subseteq \paraAC{i\uparrow}\subseteq \paraAC{i+\varepsilon}$ for any $\varepsilon>0$.
In fact, in this paper we would be able to simplify some arguments if we only wanted to prove containment in $\paraAC{1\uparrow}$.

For $i=0$, the classes $\paraAC{0}$ and $\paraAC{0\uparrow}$ are defined slightly differently: 
in $\paraAC{0}$ we require that the depth of the circuits is $\Oh(1)$, i.e.~bounded by a universal constant independent of the parameter, while in $\paraAC{0\uparrow}$ we allow the depth of $C_{n,k}$ to
be bounded by $f(k)$, for a computable function $f$.

Let us briefly elaborate on the differences between the classes $\paraAC{i}$ and $\paraAC{i\uparrow}$.
Both definitions are natural candidates for what a parameterized analogue of $\mathsf{AC}^i$ should be, as in both cases the class becomes $\mathsf{AC}^i$ whenever $k$ is fixed to be a constant.
However, bounding the depth by $f(k)+c\cdot \log n^i$ instead of $f(k)\cdot \log n^i$ gives better guarantees when transforming the circuit to a formula (a circuit with maximum fan-out $1$).
For instance, every $\paraNC{1}$ circuit (where we restrict fan-in to be at most $2$) can be unravelled to an equivalent $\paraNC{1}$ formula, which is the analogue of a well-known property of $\mathsf{NC}^1$, 
but this is no longer the case for $\paraNC{1\uparrow}$, because
the formula size would be $n^{f(k)}$ instead of $f(k)\cdot n^{\Oh(1)}$. Similarly, every $\paraAC{0}$ circuit can be unravelled to an equivalent $\paraAC{0}$ formula, but this property is not shared
by $\paraAC{0\uparrow}$. Nevertheless, classes $\paraAC{i\uparrow}$ are still worth studying due to encompassing many natural algorithms.
This state reflects the situation in parameterized analogues of nondeterministic logspace, where the difference between bounding the space by $f(k)+c\log n$ and by $f(k)\cdot \log n$ has dramatic implications for
determinization results. We refer to the work of Elberfeld et al.~\cite{ElberfeldST15} for a broader exposition of these connections.

\paragraph{Graph problems.}
Typically, throughout this paper the input to a parameterized transformation will be a graph $G$ on~$n$ vertices. In this case we will always assume that the input is encoded as the $n\times n$ binary adjacency 
matrix of the graph; thus the circuit computing the transformation needs to have~$n^2$ inputs. Abusing the above notation somewhat, for parameterized circuit classes,
we will interpret the $|x|$ for an encoding $x$ of $G$ as the number $n$ of vertices of $G$, instead of the actual length $n^2$ of $x$. Thus, the domain of a parameterized transformation defined on graphs consists
of all words of length $n^2$ for some integer $n$, interpreted as adjacency matrices, and each circuit $C_{n,k}$ will actually have $n^2$ inputs. In case the input to the transformation consists of a graph together
with some additional piece of information (e.g.~additional numerical parameters, a coloring of the graph, or an ordering of its vertices), 
the appropriate encoding of this additional information (the form of which will be specified later) 
is provided via extra input gates. In all cases, the circuit $C_{n,k}$ will be responsible for the treatment of instances with~$n$ vertices and parameter value $k$.

In case of problems parameterized by additional numerical parameters (e.g. ``given a graph~$G$ with maximum degree at most $d$, where $d$ is the parameter''), we assume that the numerical parameters are
appended to the input and the parameterization function $\kappa$ just extracts this part of the input and presents it as the parameter. Also, we do not assume that a circuit needs to check whether the
input satisfies the stated constraint. 
For instance, if we say that a circuit computes some output given a graph whose treedepth is at most $h$, then we only state that the output is computed correctly
provided the input graph has treedepth at most $h$, while we do not assert anything about the output of the circuit on graphs of higher treedepth. 


\paragraph{Basic circuit constructions.}
One of the fundamental results for parameterized circuit complexity is that counting up to a threshold parameter $d$ can be done in $\paraAC{0}$. 
More precisely, consider the problem {\sc{Threshold}} parameterized by $d$: given a word $x\in \{0,1\}^\star$, determine whether $x$ has at most $d$ ones.
A naive construction of a constant-depth circuit would be to check every $d$-tuple of inputs, but this would result in a circuit of size $\Omega(n^d)$.
However, Bannach et al.~\cite{BannachST15} showed that {\sc{Threshold}} in fact is in $\paraAC{0}$ using a parallel implementation of color coding. 

\begin{theorem}[Lemma 3.3 of~\cite{BannachST15}]\label{thm:threshold}
{\sc{Threshold}} is in $\paraAC{0}$.
\end{theorem}

Note that by our definition of $\paraAC{0}$, the circuits in the family provided by Theorem~\ref{thm:threshold} have depth bounded by a universal constant, independent of $d$.
From Theorem~\ref{thm:threshold} we can immediately derive the following corollary; henceforth, a subset of vertices of an $n$-vertex graph will be encoded in our circuits 
as its characteristic (binary) vector using $n$ gates.

\begin{corollary}\label{cor:high-deg}
The following transformation parameterized by $d$ is in $\paraAC{0}$: given a graph~$G$, compute the set of vertices of $G$ that have degree at most $d$.
\end{corollary}
\begin{proof}
For every vertex $u$ of $G$, apply the circuit given by Theorem~\ref{thm:threshold} to the row of the adjacency matrix of $G$ corresponding to $u$.
\end{proof}

Another primitive in our algorithms will be counting distances in graphs up to a 
fixed threshold $r\in \N$.
This is encapsulated in the following lemma.

\begin{lemma}\label{lem:distances}
There exists a \dlogtime-uniform family of \AC-circuits $(D_{n,r})_{n,r\in \N}$ such that each~$D_{n,r}$,
given a $n$-vertex graph~$G$, outputs the $n\times n$ boolean matrix encoding, for each pair of vertices $u,v$ of $G$, whether the distance between $u$ and $v$ in $G$ is at most $r$.
Each $D_{n,r}$ has size $n^{\Oh(1)}$ and depth $\Oh(\log r)$.
\end{lemma}
\begin{proof}
Let $A$ be the adjacency matrix of $G$ with ones added on the diagonal.
Then it suffices to compute $A^r$, the $r$th boolean power of $A$, that is, its $r$th power in the $(\OR,\AND)$-semiring over $\{0,1\}$.
Observe that the distance between $u$ and $v$ in $G$ is at most $r$ if and only if the entry of $A^r$ in the intersection of the $u$-column and $v$-row is equal to $1$.
Observe that $A^r$ can be computed from $A$ by a circuit of size $n^{\Oh(1)}$ and depth $\Oh(\log r)$ using the iterative squaring algorithm.
\end{proof}

%% file: orderings.tex
\section{Bounded degeneracy}\label{sec:orderings}

In this section we study the case of graphs of degeneracy $d$. Those are graphs which can be linearly ordered in such a way that every vertex has at most $d$ smaller neighbors. 
It is well known that a class $\CCC$ of graphs has  degeneracy
bounded by some constant $d$ if and only if there is a number $c$ such that in every subgraph $H$ of a graph $G\in\CCC$, the  ratio between the number of edges and number of vertices in $H$ is bounded by $c$.
Therefore, bounded degeneracy is similar to bounded expansion, but we only bound the density of depth-$0$ minors, i.e.~subgraphs. Many proof techniques concerning bounded expansion classes stem from the techniques for classes of bounded degeneracy. This 
is no different in our paper. In this section, we  prove some parallelized variants of known results 
for classes of bounded degeneracy.
Specifically,
it is known that  graphs of degeneracy $d$ admit a proper  coloring using $d+1$ colors, and in Lemma~\ref{lem:deg-apx-col},  we show that a  coloring using $\Oh(d^2)$ colors can be computed in $\paraAC{1}$. 
In the following section, we extend this result to classes of bounded expansion. 

%


\paragraph*{Definition and basic properties.}
A {\em{vertex ordering}} of a graph $G$ is any ordering $\sigma=(v_1,\ldots,v_n)$ of the vertices of $G$. A vertex ordering $\sigma$ can be also understood via the linear order $\leq_\sigma$ on $V(G)$ imposed by it:
$u\leq_\sigma v$ iff $u=v$ or~$u$ appears earlier than $v$ in $\sigma$. Whenever a vertex ordering of an $n$-vertex graph is represented in a circuit construction, we assume that it is represented as the 
$n\times n$ boolean matrix encoding the order $\leq_\sigma$. 
%
The {\em{degeneracy}} of the ordering~$\sigma$ is the least $d$ such that every vertex $v\in V(G)$ has at most $d$ neighbors $u$ satisfying 
$u<_\sigma v$. The {\em{degeneracy}} of a graph is the smallest possible degeneracy of its vertex orderings.

The following basic lemma is well-known. In fact, the property expressed in it is commonly used as the base definition of degeneracy.

\begin{proposition}\label{prop:subgraph}
The degeneracy of a graph $G$ is equal to the smallest integer $d$ such that every subgraph of $G$ contains a vertex of degree at most $d$.
\end{proposition}

By Proposition~\ref{prop:subgraph} it is clear that degeneracy is a monotone graph parameter, i.e., the degeneracy of a subgraph of a graph~$G$ is never larger than the degeneracy of $G$.
Further, it also yields a greedy polynomial-time algorithm for computing the degeneracy of a given graph $G$: starting with the whole graph $G$, repeatedly remove a vertex of the smallest degree from the current graph,
until there are no more vertices left. The reversal of the order of removing the vertices is then a vertex ordering of $G$ with optimum degeneracy. 
To see this, observe that suppose $d$ is the largest degree of a removed vertex
encountered during the procedure. On one hand, clearly the obtained vertex ordering has degeneracy $d$. On the other hand, at the moment when we removed a vertex of degree $d$, 
the currently considered subgraph of $G$ had minimum degree $d$, which certifies that the degeneracy of $G$ cannot be smaller than $d$ by Proposition~\ref{prop:subgraph}.

We use the following fact that a graph of degeneracy $d$ has a linear number of edges, and moreover there are few vertices with degrees significantly larger than~$d$.

\begin{proposition}\label{prop:sparse}
An $n$-vertex graph $G$ of degeneracy at most~$d$ has at most $dn$ edges. Moreover, for every real $c\geq 1$, $G$ has less than $\frac{n}{c}$ vertices of degree larger than $2cd$.
\end{proposition}
\begin{proof}
For the first assertion, consider a vertex ordering $\sigma$ of $G$ of degeneracy at most $d$ and count the edges by their higher (in $\sigma$) endpoints.
For the second assertion, observe that otherwise by the hand-shaking lemma $G$ would have more than $\frac{1}{2}\cdot \frac{n}{c}\cdot 2cd=dn$ edges, a contradiction with the
first assertion.
\end{proof}

Finally, we recall the well-known fact that a graph of degeneracy $d$ admits a proper coloring with $d+1$ colors. Recall here that a proper coloring of a graph is a coloring of its vertices such
that no edge has both endpoints of the same color. Equivalently, every color class is an independent set.

\begin{proposition}\label{prop:greedy}
A graph of degeneracy $d$ admits a proper coloring with $d+1$ colors.
\end{proposition}
\begin{proof}
Let $G$ be the graph in question and let $\sigma$ be a vertex ordering of $G$ of degeneracy $d$.
Consider the following greedy procedure that colors vertices of $G$ with colors $\{1,\ldots,d+1\}$:
iterate through vertices of $G$ in the order of $\sigma$ and for each vertex $u$ assign to it any color that is not present among the neighbors of $u$ smaller in $\sigma$.
Since there are at most $d$ such neighbors, such a color will always exist.
\end{proof}
Our goal in this section is to prove a parallelized variant of Proposition~\ref{prop:greedy}. This will be achieved in Lemma~\ref{lem:deg-apx-col} below.

\subsection{Block vertex orderings and computational aspects}
We will use a relaxed variant of degeneracy orderings where vertices come in ordered blocks and every vertex has few neighbors in its own and smaller blocks.

\begin{definition}
A {\em{block vertex ordering}} of a graph $G$ is an ordered partition $\tau=(B_1,B_2,\ldots,B_\ell)$ of the vertex set of $G$; its {\em{length}} is the number of blocks $\ell$.
The {\em{degeneracy}} of $\tau$ is the least integer~$d$ such that for each $i\in \{1,\ldots,\ell\}$, every vertex $v\in B_i$ has at most $d$ neighbors in $\bigcup_{j=1}^i B_j$.
\end{definition}

A block vertex ordering $\tau$ as above naturally imposes a total quasi-order $\leq_\tau$ on the vertex set of $G$: $u\leq_\tau v$ iff $u\in B_i$ and $v\in B_j$ with $i\leq j$.
In our circuits we will assume that a block vertex ordering of an $n$-vertex graph is represented by the $n\times n$ boolean matrix encoding $\leq_\tau$.

Obviously, if $\sigma=(v_1,\ldots,v_n)$ is a vertex ordering of $G$, then the degeneracy of $\sigma$ is equal to the degeneracy of the block vertex ordering $(\{v_1\},\ldots,\{v_n\})$.
On the other hand, if $\tau=(B_1,\ldots,B_\ell)$ is a block vertex ordering of $G$ of degeneracy~$d$, then by ordering each block arbitrarily and concatenating these orderings we obtain a vertex ordering
of $G$ of degeneracy at most $d$. 

It will be important however that provided $G$ has bounded degeneracy, we may find a block vertex ordering of small degeneracy that is of logarithmic length.
This idea is also the cornerstone of the work of Barenboim and Elkin~\cite{BarenboimE10}, who gave logarithmic-time distributed algorithms to approximately color graphs of bounded degeneracy.
Our block vertex orderings correspond to {\em{$H$-partitions}} in their nomenclature, and similarly to us they show that an $H$-partition of small degeneracy and logarithmic size can be efficiently computed
by repeatedly taking vertices of small degree. Actually, Barenboim and Elkin attribute the idea to an earlier work of Arikati et al.~\cite{ArikatiMZ97} that used the PRAM model of parallel algorithms.

\begin{lemma}\label{lem:degeneracy-ordering}
The following transformation parameterized by~$d$ is in $\paraAC{1}$: 
Given an $n$-vertex graph $G$ of degeneracy at most $d$, compute a block vertex ordering of $G$ of degeneracy at most $4d$ and length at most $\log n$.
\end{lemma}
\begin{proof}
We describe combinatorially how the block vertex ordering is constructed. 
Starting with $G_0=G$, define the last block to consist of all vertices of $G_0$ that have degree at most $4d$ in $G_0$; by Proposition~\ref{prop:sparse} with $c=2$, 
this block constitutes more than half of the vertex set of $G_0$. Remove the block from $G_0$ yielding a graph $G_1$ and apply again the same procedure to $G_1$.
That is, define the second-to-last block to consist of all vertices of $G_1$
that have degree at most $4d$ in~$G_1$ and remove this block yielding $G_2$; since $G_1$ is a subgraph of $G$, again Proposition~\ref{prop:sparse}
ensures us that in this manner we remove more than half of the remaining 
vertex set. Thus the construction finishes with an empty graph
after at most $\log n$ iterations and yielding at most $\log n$ blocks in total. It is straightforward to see that the degeneracy of the obtained block vertex ordering is
at most $4d$.

We are left with implementing the above procedure using an \AC-circuit family with prescribed size and depth constraints.
We may perform exactly $\lfloor \log n\rfloor$ iterations, where after $i$ iterations the last $i$ blocks are defined; in case the whole vertex set has already been exhausted, the next iterations are idle. 
It is straightforward to implement each iteration by an \AC-circuit of size $f(k)\cdot n^{\Oh(1)}$ and depth $\Oh(1)$ using Corollary~\ref{cor:high-deg}, 
so by performing the iterations sequentially we obtain an \AC-circuit of size $f(k)\cdot n^{\Oh(1)}$ and depth $\Oh(\log n)$.
\end{proof}

By ordering arbitrarily each block of the block vertex ordering given by Lemma~\ref{lem:degeneracy-ordering} we obtain the following corollary: given a graph $G$ of degeneracy at most $d$, computing
a vertex ordering of~$G$ of degeneracy at most $4d$ can be done in $\paraAC{1}$. In other words, this is a $4$-approximation algorithm for degeneracy in $\paraAC{1}$, where the target degeneracy is the parameter.

In fact, if in the proof of Lemma~\ref{lem:degeneracy-ordering} we replaced $c=2$ with $c=1+\varepsilon/2$ for any $\varepsilon>0$, we would still have that the procedure performs $\Oh(\log n)$ iterations while
producing a vertex ordering of degeneracy at most $2+\varepsilon$, so this constitutes a $(2+\varepsilon)$-approximation in $\paraAC{1}$. 

It is natural to ask whether the following problem of determining degeneracy exactly is also in $\paraAC{1}$: for a parameter $d$, determine whether the degeneracy of a given graph is at most $d$.
Recall here that this problem can be solved in polynomial time. We give a negative answer to this side question by proving the following theorem. 

\begin{theorem}\label{thm:degeneracy-Phard}
The following problem is $\mathsf{P}$-hard under logspace reductions: Given a graph $G$, determine whether the degeneracy of $G$ is at most $2$.
\end{theorem}

Thus, if determining degeneracy exactly was in $\paraAC{1}$, or even in $\paraAC{i}$ for any $i$, then $\mathsf{NC}=\mathsf{P}$.
Moreover, the same can be concluded about approximating degeneracy up to any factor $\alpha<\frac{3}{2}$.
Since Theorem~\ref{thm:degeneracy-Phard} is not directly relevant for our main result, but we find it interesting on its own, we include a proof in Appendix~\ref{app:degeneracy}.

\subsection{Coloring graphs of bounded degeneracy}
\label{sec:deg-proof}

Recall from Proposition~\ref{prop:greedy} that graphs of bounded degeneracy can be colored using a bounded number of colors.
In this section, we show the following, parallelized variant of this result.

\begin{lemma}\label{lem:deg-apx-col}
The following transformation parameterized by $d$ is in $\paraAC{1}$: Given a graph of degeneracy at most $d$, compute its proper coloring with $(4d+1)^2$ colors.
\end{lemma}
We outline the proof  below.
Recall that a similar statement in the context of distributed computing was obtained by Barenboim and Elkin~\cite{BarenboimE10}.
The rest of Section~\ref{sec:orderings} is devoted to outlining a  proof of Lemma~\ref{lem:deg-apx-col}. 

We shall first present how to greedily color bounded degree graphs in $\paraAC{1}$, and then leverage this understanding to color graphs of bounded
degeneracy in $\paraAC{1}$. But before this, we present a key technical lemma that will be used multiple times. Essentially it says that provided we have already achieved some proper coloring with $h$
colors, we may then use it to compute a better coloring using a circuit of depth linear in $h$.
This trick was also used by Barenboim and Elkin~\cite{BarenboimE10}.

\begin{lemma}\label{lem:degeneracy-merger}
There exists a \dlogtime-uniform family of \AC-circuits $(K_{n,d,h})_{n,d,h\in \N}$ such that each $K_{n,d,h}$,
given a $n$-vertex graph $G$ together with its block vertex ordering $\tau=(B_1,\ldots,B_\ell)$ with $\ell\leq h$ and of degeneracy $d$ with each block~$B_i$ being an independent set in $G$,
computes a proper coloring of~$G$ with $d+1$ colors. Each $K_{n,d,h}$ has size $n^{\Oh(1)}$ and depth $\Oh(h)$.
\end{lemma}
\begin{proof}
We build a coloring
$\lambda\colon V(G)\to \{1,\ldots,d+1\}$ in $h$ rounds, where in round $i$ all vertices from block $B_i$ receive their colors in $\lambda$. 
In round~$i$ every vertex $u\in B_i$ inspects all its neighbors and sets its own color to be the smallest color that is not present among its neighbors residing in lower blocks. 
Note that such a color always exists since the number of neighbors is at most $d$, by the assumption about the degeneracy of the input block vertex ordering, and there are $d+1$ colors available.
To see that $\lambda$ constructed in this way will be a proper coloring of $G$, observe that every edge $uv$ of $G$ connects two vertices from different blocks, say $u\in B_i$ and $v\in B_j$ for $i<j$.
Hence~$v$ will pick its color in $\lambda$ to be different than $\lambda(u)$.

We implement the above procedure by an \AC-circuit of polynomial size and depth $\Oh(h)$ in a natural way: the circuit consists of $h$ layers, where the $i$th layer corresponds to the $i$th iteration.
Thus, it suffices to implement the assignment of color to every vertex $u\in B_i$ using an \AC-circuit of polynomial size and constant depth.
To this end, for every color $j\in \{1,\ldots,d+1\}$ we create a circuit of constant depth that computes whether $j$ is 
present among neighbors of~$u$ from lower blocks ($B_{i'}$ for $i'<i$).
; this boils down to taking a disjunction over all vertices $v$ of 
the conjunction of the fact that $v<_\tau u$, $v$ is a neighbor of $u$, and $v$ has already received color~$j$. 
Then, the smallest color that is not present among neighbors of $u$ from lower blocks may be chosen as one that is not present, but all smaller ones are present; this
requires one additional level in the circuit.
\end{proof}

\begin{corollary}\label{cor:bnddeg-merger}
There is a \dlogtime-uniform family of \AC-circuits $(L_{n,d,h})_{n,d,h\in \N}$ such that each $L_{n,d,h}$,
given a $n$-vertex graph $G$ of maximum degree at most $d$ together with its proper coloring with $h$ colors,
computes a proper coloring of~$G$ with $d+1$ colors. Each $L_{n,d,h}$ has size $n^{\Oh(1)}$ and depth $\Oh(h)$.
\end{corollary}
\begin{proof}
If $\lambda\colon V(G)\to \{1,\ldots,h\}$ is the input proper coloring, then $(\lambda^{-1}(1),\ldots,\lambda^{-1}(h))$ is a block vertex ordering of $G$ of degeneracy at most $d$ where each block is an independent
set in $G$. Hence we may just compute this block vertex ordering using an \AC-circuit of depth $\Oh(1)$ and apply the circuit $K_{n,d,h}$ provided by Lemma~\ref{lem:degeneracy-merger}.
\end{proof}

Our first step towards the proof of  Lemma~\ref{lem:deg-apx-col} is the treatment of graphs of bounded degree.
A graph of maximum degree at most $\Delta$ can be greedily colored with $\Delta+1$ colors.
A naive implementation of this greedy procedure is sequential, but we will show now how to perform this task in $\paraAC{1}$.
We remark that finding optimum or near-optimum proper colorings of graphs of bounded maximum degree is a very classic topic in distributed computing with a vast existing literature.
We refer to the work of Barenboim~\cite{Barenboim16} for the currently fastest algorithms and an excellent overview of the area.

We first show a weaker result, namely that a proper coloring with $\Delta+1$ colors of a graph of maximum degree at most $\Delta$ can be computed in $\paraAC{1\uparrow}$, when parameterized by $\Delta$.
Recall that this means that we allow depth $f(\Delta)\cdot \log n$ instead of $f(\Delta)+\Oh(\log n)$.
This can be done using a simple Divide\&Conquer trick, which dates back to a classic $\Oh(\Delta\log n)$-time distributed algorithm for this problem of Goldberg et al.~\cite{GoldbergPS87} 
(see also~\cite{AwerbuchGLP89}).

\begin{lemma}\label{lem:col-bnddeg-pre}
The following transformation parameterized by~$\Delta$ is in $\paraAC{1\uparrow}$: Given a graph of maximum degree at most~$\Delta$, compute its proper coloring with $\Delta+1$ colors.
\end{lemma}
\begin{proof}
We perform a divide-and-conquer algorithm.
Given the input graph~$G$ with $n$ vertices, arbitrarily partition its vertex set into two subset $V_1$ and $V_2$, each of size at most~$\lceil n/2\rceil$. 
Let $G_1$ and $G_2$ be the subgraphs induced by $V_1$ and $V_2$ in $G$, respectively.
Each of $G_1,G_2$ has maximum degree at most $\Delta$, hence we can apply the algorithm recursively to both these graphs, 
yielding proper colorings $\lambda_1,\lambda_2$ of $G_1$ and $G_2$, respectively, each using $\Delta+1$ colors. By taking the union of these two colorings, where colors from different subgraphs are considered different,
we obtain a proper coloring of $G$ with $2\Delta+2$ colors. We may now apply the circuit given by Corollary~\ref{cor:bnddeg-merger} for $h=2\Delta+2$ to compute a coloring of $G$ with $\Delta+1$ colors.

To turn the above algorithm into a circuit, observe that the depth of the recursion is $\Oh(\log n)$ and computation on each level requires a circuit of polynomial size and depth $\Oh(\Delta)$, 
by Corollary~\ref{cor:bnddeg-merger}. Hence, overall the size of the circuit is polynomial and its depth is $\Oh(\Delta\log n)$.
\end{proof}

We now show containment in $\paraAC{1}$.

\begin{lemma}\label{lem:col-bnddeg}
The following transformation parameterized by~$\Delta$ is in $\paraAC{1}$: Given a graph of maximum degree at most $\Delta$, compute its proper coloring with $\Delta+1$ colors.
\end{lemma}
\begin{proof}
Let every vertex $u$ of $G$ arbitrarily (say, according to the order of inputs) put numbers $1,\ldots,\deg(u)$ on edges incident to it; 
thus every edge is labelled with two numbers, one originating from each endpoint. Such labeling can be computed by a circuit of size $f(\Delta)\cdot n^{\Oh(1)}$ and depth~$\Oh(1)$ using the circuits provided
by Theorem~\ref{thm:threshold} to count, for every neighbor $v$ of $u$, the number of neighbors of $u$ with smaller indices than $v$.

For every pair of indices $(i,j)$ with $1\leq i\leq j\leq \Delta$, let $G_{i,j}$ be the subgraph of $G$ with $V(G_{i,j})=V(G)$ and $E(G_{i,j})$ consisting of those edges $e$ of $G$, for which
one endpoint of~$e$ labelled~$e$ with $i$, and the second labelled it with $j$. Observe that the maximum degree of $G_{i,j}$ is at most $2$, since every vertex $u$ of $G$ can be adjacent to at most
two edges of $G_{i,j}$: the one it labelled with $i$ and the one it labelled with $j$. 
Using Lemma~\ref{lem:col-bnddeg-pre} we can compute, for each $1\leq i\leq j\leq \Delta$, a proper $3$-coloring $\lambda_{i,j}$ of $G_{i,j}$ using an \AC-circuit of size $n^{\Oh(1)}$ and depth $\Oh(\log n)$.
Next we construct a product coloring $\lambda$ of $G$ with $3^{\binom{\Delta+1}{2}}$ colors: a color of a vertex $u$ in $\lambda$ is the $\binom{\Delta+1}{2}$-tuple of colors of $u$ in the colorings $\lambda_{i,j}$
for all $1\leq i\leq j\leq \Delta$. Since each edge of $G$ participates in exactly one of subgraphs $G_{i,j}$, it is clear that $\lambda$ is a proper coloring of $G$.
We may finally apply Corollary~\ref{cor:bnddeg-merger} for $h=3^{\binom{\Delta+1}{2}}$ to compute a proper coloring of $G$ with $\Delta+1$ colors using an \AC-circuit of size $n^{\Oh(1)}$ and depth
$\Oh(3^{\binom{\Delta+1}{2}})$. Thus, in total we have constructed a circuit of size $f(\Delta)\cdot n^{\Oh(1)}$ and depth $\Oh(3^{\binom{\Delta+1}{2}}+\log n)$.
\end{proof}

We finally have all the tools to prove Lemma~\ref{lem:deg-apx-col}.

\begin{proof}[of Lemma~\ref{lem:deg-apx-col}]
By Lemma~\ref{lem:degeneracy-ordering} we may compute a block vertex ordering $\tau$ of $G$ of degeneracy at most $4d$ and length at most $\log n$ in $\paraAC{1}$.
Partition the edges of $G$ into two graphs $G_1,G_2$ on the same vertex set as $G$:
the edge set of~$G_1$ consists of all edges whose endpoints lie in the same block of $\tau$, while the edge set of $G_2$ consists of all edges whose endpoints lie in different blocks of $\tau$.
Observe that $G_1$ is a graph of maximum degree at most $4d$, hence we may apply Lemma~\ref{lem:col-bnddeg} to compute its proper coloring $\lambda_1$ with $4d+1$ colors in $\paraAC{1}$.
On the other hand, $\tau$ is a block vertex ordering of $G_2$ with degeneracy at most $4d$, length at most $\log n$, and every block being an independent set in $G_2$.
Hence, we may apply Lemma~\ref{lem:degeneracy-merger} to compute a proper coloring $\lambda_2$ of~$G_2$ with $4d+1$ colors using a circuit of polynomial size and depth $\Oh(\log n)$.
Finally, let $\lambda$ be the product coloring of $\lambda_1$ and $\lambda_2$: the color a vertex $u$ receives in $\lambda$ is the pair of colors it received in $\lambda_1$ and $\lambda_2$.
Since each edge of $G$ participates either in $G_1$ or in $G_2$, $\lambda$ constructed in this manner is a proper coloring of $G$ with $(4d+1)^2$ colors.
The fact that the constructed circuit satisfies the required size and depth bounds follows directly from the construction and from the bounds provided by Lemma~\ref{lem:degeneracy-ordering}, Lemma~\ref{lem:degeneracy-merger}, and
Lemma~\ref{lem:col-bnddeg}.
\end{proof}

%% file: low-td.tex
\section{Computing low treedepth colorings}\label{sec:low-td}
As discussed in the previous section,
a graph of bounded degeneracy admits a proper coloring 
using a bounded number of colors.
There is a generalization of this result 
to graphs of bounded expansion, in terms of \emph{low treedepth colorings}. Intuitively, for a fixed $p\in\N$, such a coloring is a coloring using a bounded number of colors, such that any $p$ color classes induce a graph which has a depth-first search forest of bounded depth. Such colorings turn out to be very useful for many algorithmic purposes, among others, for model-checking. In this section, we generalize the result from the previous section, and show that such colorings can be computed in $\paraAC{1}$. As previously, the 
 such colorings are obtained by first finding an appropriate ordering of the graph, related to the notion of \emph{weak reachability}, which we recall below.

\subsection{Generalized coloring numbers}

{\em{Generalized coloring numbers}} are key components of the algorithmic toolbox of the sparsity theory.
The idea is that since in classes of bounded expansion bound edge density of shallow minors at every fixed depth $r$, and bounding edge density of subgraphs corresponds to bounding degeneracy,
it is natural to generalize the notion of degeneracy to higher depth $r$ as well. 

We will use the following two generalized coloring numbers: admissibility and weak coloring number.
We prove that the generalized coloring numbers can be approximated well in $\paraAC{1}$. 

\paragraph*{Admissibility.}
Suppose $G$ is a graph and $S\subseteq V(G)$ is a subset of its vertices. The {\em{back-connectivity}} of a vertex \mbox{$u\in S$} at depth $r$ on $S$, denoted $\backconn_r(S,u)$, is 
the maximum cardinality of a family of paths $\Pp$ in $G$ with the following properties: 
\begin{itemize}
\item 
every path $P\in \Pp$ has length at most $r$, starts in $u$, ends in a vertex of $S$ different from $u$, and all its internal vertices do not belong to $S$;
\item 
every two paths from $\Pp$ share only the vertex $u$ and otherwise are vertex-disjoint.
\end{itemize}

For a vertex ordering $\sigma$ of $G$, the {\em{$r$-admissibility}} of a vertex~$v$, denoted $\adm_r(G,\sigma,v)$, is equal to $\backconn_r(\{u\colon u\leq_\sigma v\},v)$.
The {\em{$r$-admissibility}} of $\sigma$ is \[\adm_r(G,\sigma)=\max_{v\in V(G)} \adm_r(G,\sigma,v)\] and the {\em{$r$-admissibility}} of $G$, denoted $\adm_r(G)$, is the minimum possible $r$-admissibility of a vertex
ordering of $G$.

It is known that admissibility can be used to characterize classes of bounded expansion in the following sense: a class~$\Cc$ of graphs has bounded expansion if and only if there exists a function $f\colon \N\to \N$
such that $\adm_r(G)\leq f(r)$ for each $G\in \Cc$ (see e.g.~\cite{Dvorak13,grohe2015colouring,notes}). 

We use the same trick as in Lemma~\ref{lem:degeneracy-ordering} to compute a vertex ordering that achieves slightly worse $r$-admissibility, but uses a logarithmic number of rounds.
It will be convenient to think again of block vertex orderings. For a block vertex ordering $\tau=(B_1,\ldots,B_\ell)$, the $r$-admissibility of a vertex $v\in B_i$ in this ordering is defined as
$\backconn_r(\bigcup_{j=1}^i B_i,v)$, and the $r$-admissibility of~$\tau$ is defined at the maximum $r$-admissibility of any vertex in $\tau$. 
Again, it is straightforward to see that if $\tau$ is a block vertex ordering
of admissibility at most $k$, then by ordering the blocks of $\tau$ arbitrarily and concatenating these orderings as in $\tau$ we obtain a vertex ordering of $G$ of $r$-admissibility at most $k$. 

The main idea behind this result is encapsulated in the following lemma.

\begin{lemma}[explicit in Theorem~3.1 in~\cite{grohe2015colouring}]\label{lem:adm-obstacle}
Suppose $r,d\in \N$, $G$ is a graph, and $S$ is a subset of vertices of $G$ with the following property: for each $v\in S$, we have $\backconn_r(S,v)>6rd^3$.
Then $G$ admits a depth-$(r-1)$ topological minor with edge density larger than $d$.
\end{lemma}

In this work we do not use Lemma~\ref{lem:adm-obstacle} directly, but we use its stronger variant --- Lemma~\ref{lem:adm-half}, to be stated later --- which we prove explicitly in Appendix~\ref{app:adm-half}.

From Lemma~\ref{lem:adm-obstacle} it easily follows that for every graph $G$ and integer $r$, we have
$$\adm_r(G)\leq 6r(\lceil \topnabla_{r-1}(G)\rceil)^3.$$
Indeed, consider the following greedy procedure. Start with $S=V(G)$ and an empty ordering, and repeatedly perform the following until $S$ becomes empty: 
find a vertex $v\in S$ with minimum $\backconn_r(S,v)$, put $v$ at the front of the constructed ordering, and remove $v$ from $S$.
Lemma~\ref{lem:adm-obstacle} ensures us that at each step a vertex with $r$-admissibility at most $6r(\lceil \topnabla_{r-1}(G)\rceil)^3$ will be extracted, yielding the same bound on the $r$-admissibility
of the final ordering.

This approach somehow mirrors the greedy algorithm for computing the degeneracy of the graph. Observe that it performs a linear number of iterations, hence a priori it not straightforward to parallelize it.
We will now use the same trick as in Lemma~\ref{lem:degeneracy-ordering} to compute a vertex ordering that achieves slightly worse $r$-admissibility, but uses a logarithmic number of rounds.
It will be convenient to think again of block vertex orderings. For a block vertex ordering $\tau=(B_1,\ldots,B_\ell)$, the $r$-admissibility of a vertex $v\in B_i$ in this ordering is defined as
$\backconn_r(\bigcup_{j=1}^i B_i,v)$, and the $r$-admissibility of $\tau$ is defined at the maximum $r$-admissibility of any vertex in $\tau$. 
Again, it is straightforward to see that if $\tau$ is a block vertex ordering
of admissibility at most $k$, then by ordering the blocks of $\tau$ arbitrarily and concatenating these orderings as in $\tau$ we obtain a vertex ordering of $G$ of $r$-admissibility at most $k$.

\begin{lemma}\label{lem:adm-para}
The following transformation parameterized by integers $r$ and $d$ is in $\paraAC{1}$: Given a graph $G$ with $\topnabla_{r-1}(G)\leq d$, compute a vertex block ordering of $G$ with $r$-admissibility at most
$6r^2d^3$ and length at most $\log n$.
\end{lemma}

For the proof of Lemma~\ref{lem:adm-para} we need analogues of the two ingredients that we used in the proof of Lemma~\ref{lem:degeneracy-ordering}. First, we need to know that at every iteration
a vast majority of the remaining vertices can be removed and set to be the next block. This is done by the following lemma.

\begin{lemma}\label{lem:adm-half}
Suppose $\varepsilon>0$, $r,d\in \N$, $G$ is a graph, and $S$ is a subset of vertices of $G$ with the following property: 
for at least $\varepsilon|S|$ vertices $v\in S$ we have $\backconn_r(S,v)>6\varepsilon^{-1}rd^3$.
Then $G$ admits a depth-$(r-1)$ topological minor with edge density larger than $d$.
\end{lemma}

The proof of Lemma~\ref{lem:adm-half} follows closely the lines of the proof of Lemma~\ref{lem:adm-obstacle}, however we need to argue that having large back-connectivity of at least $\varepsilon$-fraction of
vertices of $S$, instead of
all, is sufficient to construct a depth-$(r-1)$ topological minor model of a graph with edge density larger than $d$. This essentially requires modification of numerical parameters in the proof.
For the sake of completeness we include the proof of Lemma~\ref{lem:adm-half} in Appendix~\ref{app:adm-half}.

The second necessary ingredient is that we need to compute the set of vertices with low back-connectivity efficiently, similarly as in the proof of Lemma~\ref{lem:degeneracy-ordering} it was necessary to compute the
set of vertices with low degree in the remaining graph using Corollary~\ref{cor:high-deg}. For this we use the following lemma, whose proof, similarly as that of Corollary~\ref{cor:high-deg}, relies on color coding.

\begin{lemma}\label{lem:backconn}
Consider the following problem parameterized by $r$ and $k$: Given a graph $G$, its vertex subset $S$, and a vertex $u\in S$, determine whether $\backconn_r(S,u)<k$.
Then this problem can be solved by a family $(A_{n,r,k})_{n,r,k\in \N}$ of \AC-circuits where each $A_{n,r,k}$ has depth $\Oh(\log r)$ and size $f(r,k)\cdot n^{\Oh(1)}$, for some computable function $f$.
\end{lemma}
\begin{proof}
We first recall the toolbox used by Bannach et al.~\cite{BannachST15}.

\begin{definition}[Universal coloring family]
For integers $n,k,c$, an {\em{$(n,k,c)$-universal coloring family}} is a family $\Lambda$ of functions from $[n]$ to $[c]$ with the following property:
for every subset $S\subseteq [n]$ with $|S|\leq k$ and function $f\colon S\to [c]$, there exists $\lambda\in \Lambda$ such that the restriction of~$\lambda$ to~$S$ is equal to $f$. 
\end{definition}

\begin{theorem}[Theorem 3.2 of~\cite{BannachST15}]\label{thm:ucf}
There is a \dlogtime-uniform family of \AC-circuits $(C_{n,k,c})_{n,k,c\in \N}$ without inputs such that each $C_{n,k,c}$:
\begin{itemize}
\item has depth $\Oh(1)$ and size $f(k,c)\cdot n^{\Oh(1)}$ for a computable function $f$; and
\item outputs an $(n,k,c)$-universal coloring family, encoded as a list of function tables.
\end{itemize}
\end{theorem}

Suppose $\Lambda$ is a $(n,rk,k)$-universal coloring family on the vertex set of $G$ (identified with $[n]$).
Provided $\backconn_r(S,u)\geq k$ there exists a path family $\Pp$ witnessing this fact: $|\Pp|=k$, each $P\in \Pp$ has length at most $r$ and leads from $u$ to another vertex of $S$ through vertices
outside of $S$, and paths from $\Pp$ pairwise share only $u$. 
Letting $\Pp=\{P_1,\ldots,P_k\}$ we have $|\bigcup_{i=1}^k (V(P_i)-\{u\})|\leq rk$, hence there exists a coloring $\lambda\in \Lambda$ such that
for each $i\in [k]$, all the vertices of $V(P_i)-\{u\}$ are colored with color $i$ in $\lambda$. We shall then say that $\Pp$ is {\em{well-colored}} by~$\lambda$.

This suggest the following construction of $A_{n,r,k}$.
First, apply the circuit $C_{n,rk,k}$ given by Theorem~\ref{thm:ucf} to construct an $(n,rk,k)$-universal coloring family $\Lambda$ on the vertex set of $G$.
Then, for each $\lambda\in \Lambda$ construct a circuit verifying whether there is a path family $\Pp$ as above that is well-colored by $\lambda$.
For this, for each $i\in [k]$ construct the graph $G_i$ which is the subgraph of $G$ induced by  $u$ and all vertices of $G$ that are colored with color $i$ in $\lambda$.
Then check whether any vertex of $S$ is at distance at most $r$ from $u$ in $G_i$ using the circuit provided by~Lemma~\ref{lem:distances}; this circuit has size $n^{\Oh(1)}$ and depth $\Oh(\log r)$.
The discussion from the previous paragraph proves that $\backconn_r(S,u)<k$ if and only if for every $\lambda\in \Lambda$ the circuit constructed for $\lambda$ 
did not succeed in finding a path family well-colored by $\lambda$.
\end{proof}

Observe that Lemma~\ref{lem:adm-half} and Lemma~\ref{lem:backconn} do not fit so nicely together for the proof of Lemma~\ref{lem:adm-para} as it was the case for Lemma~\ref{lem:degeneracy-ordering}.
The reason is that the circuits provided by Lemma~\ref{lem:backconn} have depth $\Oh(\log r)$ instead of $\Oh(1)$, so applying exactly the same strategy --- of removing at each step half of vertices --- would
yield a circuit of depth $\Oh(\log r\cdot \log n)$, which is too much for $\paraAC{1}$. The idea is to put $\varepsilon=1/r$ instead of $\varepsilon=1/2$ in order to trade the approximation factor for the
depth of the circuit. Thus the iteration has length $\Oh(\log n/\log r)$,
so in the final circuit we need $\Oh(\log n/\log r)$ layers of depth $\Oh(\log r)$ each, and the unwanted term $\log r$ cancels out. We now provide formal details.

\begin{proof}[of Lemma~\ref{lem:adm-para}]
Consider the following iterative procedure. Starting with $S=V(G)$ and empty block vertex ordering, 
repeatedly find the set of those vertices $v\in S$ for which $\backconn_r(S,v)>6r^2d^3$, remove it from $S$ and put it as the next block at the front of the constructed block vertex ordering.
By Lemma~\ref{lem:adm-half}, each iteration results in decreasing the size of $S$ by a multiplicative factor larger than $r$, thus the iteration terminates yielding a block vertex ordering of the whole graph within
at most $\log_r n= \frac{\log n}{\log r}\leq \log n$ iterations. Thus, it is straightforward to see that the obtained block vertex order has length at most $\log n$ and 
$r$-admissibility at most $6r^2d^3$, as requested. For the implementation using a circuit family, by Lemma~\ref{lem:backconn} each iteration can be performed using a circuit of size $f(r,d)\cdot n^{\Oh(1)}$
and depth $\Oh(\log r)$, and we may apply $\lfloor\frac{\log n}{\log r}\rfloor$ iterations sequentially. Thus the obtained circuit has size $f(r,d)\cdot n^{\Oh(1)}$ and depth $\Oh(\log n)$, as claimed.
\end{proof}

\paragraph*{Weak coloring number.}
Suppose $G$ is a graph, $r\in \N$, and~$\sigma$ is a vertex ordering of~$G$. For two vertices $u,v\in V(G)$ with $u\leq_\sigma v$, we say that $u$ is {\em{weakly $r$-reachable}} from $v$
if there is a path of length $r$ in~$G$ that leads from $v$ to $u$ and whose all internal vertices are larger than $u$ in~$\sigma$. 
The set of vertices weakly $r$-reachable from $v$ in $\sigma$ is denoted by $\WReach_r[G,\sigma,v]$.
The {\em{weak $r$-coloring number}} of $\sigma$ is equal to $$\wcol_r(G,\sigma)=\max_{v\in V(G)}|\WReach_r[G,\sigma,v]|$$ and the weak $r$-coloring number of $G$, denoted $\wcol(G)$, is the smallest weak $r$-coloring
number of a vertex ordering of $G$.

While measuring the complexity of vertex orderings via weak coloring numbers is arguably more useful than via admissibility,
it turns out that $r$-admissibility and weak $r$-coloring number are functionally equivalent.
While it is straightforward that $\adm_r(G,\sigma)\leq \wcol_r(G,\sigma)$ for any \mbox{$r\in \N$}, graph $G$, and its vertex ordering $\sigma$,
the weak coloring number is also bounded from above by a function of admissibility as follows:

\begin{lemma}[Theorem 2.6 of \cite{Dvorak13}]\label{lem:adm-wcol}
For any $c,r\in \N$ with \mbox{$c\geq 2$}, graph $G$, and vertex ordering~$\sigma$ of $G$ with \mbox{$\adm_r(G,\sigma)\leq c$}, we have 
$\wcol_r(G,\sigma)\leq \frac{c^{r+1}-1}{c-1}.$
\end{lemma}

Thus, by setting
$g(r,d)\coloneqq \frac{(6r^2d^3)^{r+1}-1}{6r^2d^3-1}$
and combining Lemma~\ref{lem:adm-para} with Lemma~\ref{lem:adm-wcol}, we obtain the following. 

\begin{theorem}\label{thm:wcol-apx}
The following transformation parameterized by integers $r$ and $d$ is in $\paraAC{1}$: 
Given a graph $G$ with $\nabla_{r-1}(G)\leq d$, compute a vertex ordering of $G$ with $\wcol_r(G)\leq g(r,d)$.
\end{theorem}

Our approach to proving Theorem~\ref{thm:wcol-apx} is as follows. The weak 
coloring numbers are closely related to another measure called admissibility.
More precisely, every order witnessing that the
$r$-admissibility is small also witnesses that the weak $r$-coloring number is small. 
For computing $r$-ad\-missi\-bility there exists a greedy approximation 
algorithm~\cite{dvovrak2013constant,grohe2015colouring}. We turn this 
greedy approximation algorithm into a low-depth circuit using a similar trick as 
we did for degeneracy. In each single step of the algorithm we 
make use of the color coding toolbox provided by Bannach et al.~\cite{BannachST15}.

For applications we will need to efficiently compute the weak $r$-reachability relation, as expressed in the next lemma.

\begin{lemma}\label{lem:wreach}
Consider the following problem parameterized by $r$: Given a graph $G$, its vertex ordering $\sigma$, and two vertices $u$ and $v$, determine whether $u$ is weakly $r$-reachable from $v$ in $\sigma$.
Then this problem can be solved by a \dlogtime-uniform family $(W_{n,r})_{n,r\in \N}$ of \AC-circuits where each $W_{n,r}$ has size $n^{\Oh(1)}$ depth $\Oh(\log r)$.
\end{lemma}
\begin{proof}
It suffices to verify whether $u\leq_\sigma v$ and the distance between $u$ and $v$ in the subgraph induced by vertices not smaller in $\sigma$ than $u$ is at most $r$.
The former can be read from an input gate, while the latter can be done using a circuit of size $n^{\Oh(1)}$ and depth $\Oh(\log r)$ by Lemma~\ref{lem:distances}.
\end{proof}

Actually, at this point we may already prove Theorem~\ref{thm:domset-main}. Amiri et al.~\cite{AmiriMRS17} observed that given a vertex ordering with a low weak $2r$-coloring number, an approximate distance-$r$ dominating set
can be found using a very simple selection rule.

\begin{theorem}[explicit in the proof of Theorem~8 of~\cite{AmiriMRS17}]\label{thm:distributed}
Suppose $r\in \N$, $G$ is a graph, and $\sigma$ is a vertex ordering of $G$.
For $u\in V(G)$, let $d(u)$ be the vertex of $\WReach_r[G,\sigma,u]$ that is the smallest in the ordering $\sigma$, and define $D\coloneqq \{d(u)\colon u\in V(G)\}$.
Then $D$ is a distance-$r$ dominating set in $G$ and $|D|\leq \wcol_{2r}(G,\sigma)\cdot \dom_r(G)$.
\end{theorem}

\begin{proof}[of Theorem~\ref{thm:domset-main}]
Since $\Cc$ has effectively bounded expansion, for every $r\in \N$ there exists a constant $d\in \N$, computable from $r$,
such that no graph from $\Cc$ admits a depth-$(2r-1)$ topological minor with edge density larger than $d$.
Consequently, given an $n$-vertex graph $G\in \Cc$ we may apply the circuit provided by Theorem~\ref{thm:wcol-apx} to compute a vertex ordering $\sigma$ of $G$ with $\wcol_{2r}(G,\sigma)\leq g(2r,d)$, where the latter
is a constant depending only on $r$ in a computable manner. This circuit has size $f(r,d)\cdot n^{\Oh(1)}$ and depth $\Oh(\log n)$, for computable $f$.
Next, by Lemma~\ref{lem:wreach} we may compute, for every pair of vertices $u,v\in V(G)$, whether $u\in \WReach_r[G,\sigma,v]$ using a circuit of size $n^{\Oh(1)}$ and depth $\Oh(\log r)\leq \Oh(\log n)$.
Finally, for every vertex $u\in V(G)$ we may construct a circuit of depth $\Oh(1)$ that finds $d(u)$.
We conclude by adding one output gate for each $v\in V(G)$ that computes whether $v$ is to be included into $D$ by checking whether $v=d(u)$ for any $u\in V(G)$.
It follows directly from the construction that the circuit satisfies the required size and depth constraints, while the correctness of the output is asserted by Theorem~\ref{thm:distributed}.
\end{proof}

\subsection{Low treedepth colorings}

Being able to efficiently compute vertex orderings with low weak coloring numbers enables us to compute low treedepth colorings.
Let us now introduce the relevant definitions.

A {\em{separation forest}}\footnote{This notion is also called {\em{elimination forest}} in the literature; we find the name {\em{separation forest}} more~explanatory.} of a graph $G$ is a forest $F$ on the same vertex set as $G$ such that whenever $uv$ is an edge in $G$, then either $u$
is an ancestor of $v$, or $v$ is an ancestor of $u$ in $F$
The {\em{treedepth}} of a graph $G$ is the smallest possible depth of a separation forest of $G$.
For an integer~$p$, a coloring $\lambda\colon V(G)\to \{1,\ldots,M\}$ is a {\em{treedepth-$p$ coloring}} of $G$ if every $i$-tuple of color classes in $\lambda$, $i\leq p$, induces
in $G$ a graph of treedepth at most $i$.

It is shown in~\cite{nevsetvril2008grad}
that a class $\CCC$ of graphs has bounded expansion if and only if for every $p$ there is a number $M$
such that every graph $G\in \CCC$ admits a treedepth-$p$
coloring using $M$ colors. 
We remark that the above definition of a treedepth-$p$ coloring can be relaxed, yielding a less restrictive definition that is  sufficient for most algorithmic purposes, including our purposes in this paper. 
Namely, it would be sufficient to require that every $p$-tuple
of classes induces in $G$ a graph of treedepth at most $f(p)$,
for some function $f\from\N\to\N$. It follows from the proof in~\cite{nevsetvril2008grad} that this weaker variant 
yields a notion 
that is still equivalent to having bounded expansion. 
Below, we use the original notion of treedepth-$p$ colorings.

\medskip
We now show how to compute low treedepth colorings orderings with low weak $r$-coloring numbers.
Suppose $G$ is a graph and $\sigma$ is a vertex ordering of $G$. For $r\in \N$, let $G\langle r,\sigma\rangle$ be the {\em{weak $r$-reachability graph}} of $\sigma$, whose vertex set is $V(G)$ and
where $u<_\sigma v$ are considered adjacent if and only if $u\in \WReach_r[G,\sigma,v]$. The following lemma explains the relation between weak $r$-reachability graphs and low treedepth colorings.
\footnote{
To be more precise, (the proof of) Theorem 2.6 of \cite{zhu2009colouring} asserts that the proper coloring of $G\langle 2^{p-2},\sigma\rangle$ obtained by the greedy coloring procedure of Proposition~\ref{prop:greedy}
is a {\em{$p$-centered coloring}} of $G$.
The fact that the coloring is obtained by the greedy procedure is irrelevant in the proof, the reasoning works for any proper coloring of $G\langle 2^{p-2},\sigma\rangle$.
Also, the notion of a $p$-centered coloring is actually stronger than that of a treedepth-$p$ colorings: every $p$-centered coloring is in particular a treedepth-$p$ coloring.
We invite the reader to~\cite[Chapter 2]{notes} for a comprehensive presentation of these results.}

\begin{lemma}[implicit in Theorem 2.6 of \cite{zhu2009colouring}]\label{lem:wcol-ltd}
For any graph~$G$, its vertex ordering $\sigma$, and integer $p\in \N$, every proper coloring of $G\langle 2^{p-2},\sigma\rangle$ is a treedepth-$p$ coloring of $G$.
\end{lemma}

Observe that if $G$ is a graph with a vertex ordering $\sigma$ such that $\wcol_{2^{p-2}}(G,\sigma)\leq c$ for some $c\in \N$, then
$\sigma$ is actually a vertex ordering of $G\langle 2^{p-2},\sigma\rangle$ of degeneracy $c-1$. This implies that $G\langle 2^{p-2},\sigma\rangle$ admits a proper coloring with $c$ colors.
We already know how to efficiently compute vertex orderings with low weak $r$-coloring numbers for graphs from any class of bounded expansion, see Theorem~\ref{thm:wcol-apx}.
Applying Lemma~\ref{lem:deg-apx-col} to the graph $G\langle 2^{p-2},\sigma\rangle$, we get a parallelized algorithm for computing a treedepth-$p$ coloring of a given graph $G$.
The main part of this section will be devoted to the proof of the following lemma.
From now we assume that a coloring of an $n$-vertex graph with $M$ colors is represented by $Mn$ gates, $M$ for each vertex $u$, where in an encoding of a coloring exactly one of these gates is set to
true and this gate denotes the color of~$u$.
We get the following result as an immediate corollary.

\begin{theorem}\label{thm:lowtd-colorings}
Suppose $\Cc$ is a class of effectively bounded expansion.
Then for every $p\in \N$ there exists a constant $M=M(p)$, computable from $p$, such that the following transformation parameterized by $p$ is in $\paraAC{1}$:
given a graph $G\in \Cc$, compute its treedepth-$p$ coloring using~$M$ colors.
\end{theorem}
\begin{proof}
Let $r=2^{p-2}$. Since $\Cc$ has effectively bounded expansion, there exists a constant $d\in \N$, computable from $p$,
such that no graph from $\Cc$ admits a depth-$(r-1)$ topological minor with edge density larger than $d$.
Consequently, given $G\in \Cc$ we may apply the circuit provided by Theorem~\ref{thm:wcol-apx} to compute a vertex ordering $\sigma$ of $G$ with $\wcol_r(G,\sigma)\leq g(r,d)$, where the latter
is again a constant computable from $p$. This circuit has size $f(p,d)\cdot n^{\Oh(1)}$ and depth $\Oh(\log n)$, for computable $f$.
Next, by applying the circuit provided by Lemma~\ref{lem:wreach} to each pair of vertices in~$G$ we may compute the adjacency matrix of the graph $G\langle r,\sigma\rangle$.
Since $\sigma$ witnesses that $G\langle r,\sigma\rangle$ is $(g(r,d)-1)$-degenerate, by applying the circuit provided by Lemma~\ref{lem:deg-apx-col} to $G\langle r,\sigma\rangle$
we may compute a proper coloring of this graph using at most $M\coloneqq (4g(r,d))^2$ colors, which is a constant depending on $p$ in a computable manner.
Lemma~\ref{lem:wcol-ltd} asserts that this coloring is a treedepth-$p$ coloring of $G$.
The claimed size and depth bounds on the obtained circuit follow directly from the construction and from bounds provided by Theorem~\ref{thm:wcol-apx}, Lemma~\ref{lem:wreach}, and Lemma~\ref{lem:deg-apx-col}.
\end{proof}

We remark that the problem of computing a low treedepth coloring can be also approached using {\em fraternal augmentations},
as was done e.g. in~\cite{dvovrak2013testing,nevsetvril2008grad,NesetrilM16}.
Both in our line of reasoning and in this approach the key step is computing a vertex
 ordering of low degeneracy, that is, Lemma~\ref{lem:degeneracy-ordering}.



\subsection{Computing separation forests}

A low treedepth coloring is still not enough for the model-checking algorithm to work, as we also need to compute separation forests witnessing that appropriate induced subgraphs have bounded
treedepth. In general computing separation forests of optimal depth is a hard computational problem, but if one allows approximate depth
there is a very simple and well-known way to do it
(see Section 17.3 in~\cite{sparsity}). 
The first observation is that graphs of bounded treedepth do not contain long paths.

\begin{lemma}[see Section 6.2 in~\cite{sparsity}]\label{lem:td-path}
A path on $2^k$ vertices has treedepth $k+1$.
\end{lemma}

Since treedepth is a monotone parameter, i.e. can only decrease under taking a subgraph, we immediately obtain the following statement.

\begin{corollary}\label{cor:small-diam}
A graph of treedepth at most $h$ does not contain a path on $2^h$ vertices as a subgraph. 
Consequently, in a graph of treedepth $h$ every connected component has diameter smaller than $2^h$.
\end{corollary}

Given a graph $G$, a {\em{DFS forest}} of $G$ is any separation forest~$F$ of $G$ that has the following property: whenever $u$ is a parent of $v$ in $F$, then there is an edge $uv$ in $G$.
The name is derived from the easy fact that any forest constructed by subsequent recursive calls of a depth-first search in $G$ is a separation forest of $G$ with this property.
If such a DFS forest $F$ had depth at least $2^h$, then this would witness that $G$ contains a path on $2^h$ vertices as a subgraph, 
and consequently, by Lemma~\ref{lem:td-path}, the treedepth of $G$ would be larger than $h$.
Therefore, we have the following.

\begin{lemma}\label{lem:DFS-sepfor}
Any DFS forest of a graph $G$ of treedepth at most $h$ is a separation forest of $G$ of depth smaller than $2^h$.
\end{lemma}

Lemma~\ref{lem:DFS-sepfor} gives a very simple linear-time algorithm to compute a bounded-depth separation forest of a graph of bounded treedepth --- just run depth-first search on the graph and output 
the obtained DFS forest. However, depth-first search is an inherently sequential algorithm.
As shown by Bannach et al.~\cite[Lemma 6]{BannachT16}, on graphs of bounded treedepth a DFS forest can be computed in $\paraAC{0\uparrow}$  parameterized by treedepth.

We now reprove this result for completeness.
Henceforth, every separation forest $F$ on $n$ vertices, in particular every DFS forest, 
will be represented in our circuits by its parent relation, encoded as a boolean $n\times n$ matrix.

\begin{lemma}[see also Lemma 6 of \cite{BannachT16}]\label{lem:DFS-compute}
The following transformation parameterized by $h$ is in $\paraAC{0\uparrow}$: Given a graph of treedepth $h$, compute any its DFS forest.
\end{lemma}
\begin{proof}
We first observe the following direct consequence of Corollary~\ref{cor:small-diam}.

\begin{claim}\label{cl:connectivity}
The following problem parameterized by $h$ is in $\paraAC{0\uparrow}$: 
Given an $n$-vertex graph~$G$ of treedepth at most $h$ and two vertices $u,v\in V(G)$, determine whether $u$ and $v$ are in the same connected component of $G$.
\end{claim}
\begin{proof}\renewcommand\proofSymbol{\ensuremath{\lrcorner}}
By Corollary~\ref{cor:small-diam}, to check whether $u$ and $v$ are in the same connected component of $G$ it suffices to check whether they are at distance at most $2^h$ in $G$.
By Lemma~\ref{lem:distances}, this can be done by an \AC-circuit of size $n^{\Oh(1)}$ and depth $\Oh(h)$. 
\end{proof}

With this statement in mind, we proceed with the proof. Let $G$ be the input graph; we assume that there is an implicit order $\preceq$ on the vertices of $G$, say imposed by the order of inputs.

First, let $R$ be the set consisting of the $\preceq$-smallest vertex from each connected component of $G$.
Observe that $R$ can be computed as follows: 
a vertex $u$ belongs to $R$ if for every other vertex $v$, it is not true that $v\prec u$ and $u,v$ are in the same
connected component of $G$; the latter check can be done in $\paraAC{0\uparrow}$ by Claim~\ref{cl:connectivity}. 
The set $R$ constitutes the roots of DFS trees in the connected components of $G$.

We then proceed in $2^h-1$ rounds, where in the $i$th round we construct the $(i+1)$st level of the DFS forest.
We maintain vertex subsets $X\subseteq Z$, initially both set to $R$, with the following meaning: $Z$ is the set of vertices already placed in the constructed forest (i.e. in the $i$th round $Z$
comprises vertices from levels $1,\ldots,i$ of the forest), and $X$ is the set of vertices placed in the constructed forest in the previous round (i.e. in the $i$th round $X$ comprises vertices from level $i$).
In this round we will compute the set of vertices contained in the next, $(i+1)$st, level of the DFS forest, and for each of them we will compute its parent in the DFS forest.
Transforming this encoding to the assumed encoding of separation forests via the ancestor-descendant relation will be done at the end.

To compute the next level, we perform the following steps.
First, compute the set of all vertices that are not in $Z$, but have a neighbor in $X$; call it $M$.
Also, for every vertex in $M$ record its $\preceq$-smallest neighbor in $X$. 
Second, let $Y$ be the set of $\preceq$-minimal vertices of $M$ in their respective connected components of $G-Z$; 
that is, a vertex $v\in M$ belongs to $Y$ if and only if there is no vertex $u\in M$ such that $u\prec v$
and $u,v$ are in the same connected component of $G-Z$.
Finally, we set $Y$ to be the next level of the DFS forest: put $Z\coloneqq Z\cup Y$ and $X\coloneqq Y$, and for each $u\in Y$ set its recorded neighbor in $X$ to be its parent in the DFS forest.

It is straightforward to see that the above procedure correctly computes a DFS forest of $G$.
By Lemma~\ref{lem:DFS-sepfor}, any DFS forest of $G$ has depth smaller than $2^h$, hence performing $2^h-1$ rounds suffices to exhaust the whole vertex set.
As for the implementation, it is easy to perform each round in $\paraAC{0\uparrow}$ using Claim~\ref{cl:connectivity}, and since the number of rounds is bounded by a computable function of $h$,
it follows that the overall circuit family satisfies size and depth bounds for $\paraAC{0\uparrow}$.
\end{proof}

Finally, we need to change the encoding of the separation forest, from the child-parent relation to ancestor-descendant relation.
This requires computing the transitive closure of the child-parent relation.
Since the computed rooted forest has depth less than $2^h$, it suffices to compute the $2^h$ boolean power of the matrix of the child-parent relation (treated as a directed graph), which boils down to
squaring this matrix $h$ times. Hence, this can be done by a circuit of polynomial size and depth $\Oh(h)$.

%% file: applications.tex
\section{Model checking}\label{sec:mc}

In this section we prove our main result, Theorem~\ref{thm:main}, 
and along the lines reprove the result of Dvo\v{r}\'ak et al.~\cite{dvovrak2013testing} that model-checking $\FO$ on classes of effectively
bounded expansion is in~$\linFPT$.

The general idea of our proof 
is as follows. We prove the existence of a certain efficient quantifier-elimination procedure for classes of bounded expansion. This is first done in the case of forests of bounded depth. This lifts immediately to classes of bounded treedepth. Finally, this lifts to classes of bounded expansion, via low treedepth colorings.

\subsection{Preliminaries on relational structures and logic}

We assume familiarity with basic notation for relational structures.
A {\em{vocabulary}} is a finite relational signature $\Sigma$ consisting of relation names. Each relation name $R\in \Sigma$ has
a prescribed arity $\ar(R)$, which is a positive integer --- that is, we do not allow $0$-ary relations.
The {\em{arity}} of a vocabulary~$\Sigma$ is the maximum arity of a relation in $\Sigma$.
Throughout this section all vocabularies will be of arity at most $2$; that is, they consist only of unary and binary relations.
Unary relations will be also called {\em{labels}} for brevity.

For a vocabulary $\Sigma$, a {\em{$\Sigma$-structure}} $\As$ consists of a nonempty universe $V(\As)$ and, for each relation name $R\in \Sigma$, its {\em{interpretation}} $R^\As\subseteq V(\As)^{\ar(R)}$.
If the structure $\As$ is clear from the context we may drop the superscript, thus identifying a relation name with its interpretation in the structure.
The {\em{size}} of a structure $\As$, denoted $|\As|$, is the cardinality of its universe. 
The {\em{Gaifman graph}} of $\As$, denoted $G(\As)$, has $V(\As)$ as the vertex set, 
and we make two distinct elements $u,v$ adjacent in $G(\As)$ iff $u$ and $v$ appear together in some relation in $\As$.

In our circuit constructions any $\Sigma$-structure $\As$ on~$n$ elements will be always encoded using a boolean table of length~$n$ for each unary relation in $\Sigma$
and a boolean $n\times n$ matrix for each binary relation in $\Sigma$. Similarly as for graphs, in our circuit families we create one circuit $C_{n,k}$ responsible for
tackling instances with $n$ elements and parameter value $k$. The input to $C_{n,k}$ thus consists of $n|\Sigma_1|+n^2|\Sigma_2|$ gates, where $\Sigma_1,\Sigma_2$ respectively denote the sets
of unary and binary relations of the vocabulary $\Sigma$.

For brevity we often write $\bar x$ as a shorthand for a tuple $(x_1,\ldots,x_k)$ of variables, denoting the length of the tuple by $|\bar x|\coloneqq k$.
Notation $\exists_{\bar x}$ is a shorthand for $\exists_{x_1}\exists_{x_2}\ldots\exists_{x_k}$.
For a formula $\varphi(x_1,\ldots,x_k)\in \FO[\Sigma]$, by $\varphi(\As)$ we denote the set of all tuples $(u_1,\ldots,u_k)\in V(\As)^k$
for which $\As\models \varphi(u_1,\ldots,u_k)$.
Formulas $\varphi(\bar x)$ and $\varphi'(\bar x)$ are {\em{equivalent over $\As$}} if $\varphi(\As)=\varphi'(\As)$, and simply {\em{equivalent}} if they are equivalent over every $\Sigma$-structure $\As$.

A formula $\varphi(\bar x)$ is {\em{quantifier-free}} if it has no quantifiers, i.e., it is a boolean combination of relation checks on the free variables from $\bar x$.
Further, $\varphi(\bar x)$ is {\em{existential}} if it is a positive boolean combination (i.e.\ without negations) of formulas in {\em{prenex existential form}}
$\exists_{\bar y}\ \psi(\bar x,\bar y)$,
where $\psi$ is quantifier-free. Existential formulas are closed under positive boolean combinations.
 directly from the definition, but it can be easily seen that every existential formula
can be reduced to prenex existential form.

\begin{proposition}\label{prop:prenex}
For every existential formula $\varphi(\bar x)\in \FO[\Sigma]$ there exists an equivalent formula $\varphi'(\bar x)\in \FO[\Sigma]$ in the prenex existential form, computable from $\varphi$.
\end{proposition}

\subsection{Quantifier elimination on forests of bounded depth}\label{sec:trees}

In this section we work out a basic primitive for our quantifier elimination procedure, namely the case of unordered, labeled forests of bounded depth.

\paragraph*{Rooted forests.}
Suppose $\Ts$ is a rooted, unordered forest, so far without any labels. 
We use standard notions like parent, child, ancestor, descendant, and we follow the convention that each node is regarded as its own ancestor and descendant.
The size $|\Ts|$ of a forest $\Ts$ is the number of nodes in it.
The {\em{depth}} of a node $x$ is the number of its ancestors, and the {\em{depth}} of a forest $\Ts$ is the largest depth of a node in $\Ts$.

Throughout this section we work with forests of depth at most $d$, for a fixed constant $d\in \N$.
A forest $\Ts$ of depth at most~$d$ will be modeled as a relational structure whose universe is the node set and there is one binary relation $\parent$,
where $\parent(x,y)$ holds if $x$ is the parent of $y$. 
That is, such a $\{\parent\}$-structure models a forest of depth at most~$d$ if the following conditions hold: 
each node has at most one parent and after following the parent relation from any node we always reach a node with no parent within at most $d-1$ steps.

\medskip
The following formulas will play a key role in our reasonings.

\begin{proposition}\label{prop:lcd}
There exist existential formulas 
\[\lcd_0(x,y),\lcd_1(x,y),\ldots,\lcd_d(x,y)\in \FO[\parent],\]
each of quantifier depth at most $2d$, such that for every forest~$\Ts$ of depth at most $d$, $i\in [0,d]$, and nodes $x,y$ we have
$\Ts\models\lcd_i(x,y)$ if and only if $x$ and $y$ have exactly $i$ common ancestors in $\Ts$.
\end{proposition}
\begin{proof}
Make a disjunction over possible depths $i_x$ of $x$ and $i_y$ of $y$. 
For fixed $(i_x,i_y)$, quantify existentially the $i_x$ ancestors of $x$ and $i_y$ ancestors of $y$, and check that exactly the first $i$ of them are equal to each other.
\end{proof}

Observe that the condition in Proposition~\ref{prop:lcd} can be equivalently stated as follows: $\lcd_0(x,y)$ holds iff $x$ and $y$ have no common ancestor (i.e.\ they reside in different trees of the forest), 
and for $i\geq 1$ $\lcd_i(x,y)$ holds iff
the least common ancestor of $x$ and $y$ is at depth $i$ in $\Ts$.
Note that for a node $x$, the formula $\lcd_i(x,x)$ holds if and only if $x$ is at depth $i$.
Moreover, the condition that $x$ is an ancestor of $y$ can be expressed as follows: $\lcd_i(x,y)$ holds iff $\lcd_i(x,x)$ holds, for all $i\in [0,d]$.
Thus the formulas $\lcd_i$ can be used to check the depths of nodes and the ancestor relation (using boolean combinations).

For a forest $\Ts$ and a finite label set $\Lambda$, a {\em{$\Lambda$-labeling}} of $\Ts$ is any structure obtained from $\Ts$ by
adding a unary relation $c$ for each label $c\in \Lambda$. 
A $\Lambda$-labeling of a forest will be also called a {\em{$\Lambda$-labeled forest}}.
For another label set $\wLambda$, a {\em{$\wLambda$-relabeling}} of a $\Lambda$-labeled forest~$\Ts$ is any $\wLambda$-labeled forest~$\Ss$ such that
the underlying unlabeled forests of $\Ss$ and~$\Ts$ are equal; that is, $\Ss$ is obtained from $\Ts$ by clearing all labels from~$\Lambda$ and adding new labels from $\wLambda$ in an arbitrary way.
We may drop the label set used in a labeled forest if we do not wish to specify it.

\paragraph*{Lcd types.}
Fix a label set $\Lambda$; we consider $\Lambda$-labeled forests of depth at most $d$.
A formula $\psi(\bar x)\in \FO[\{\parent\}\cup\Lambda]$ is {\em{lcd-reduced}} 
if it uses neither quantifiers nor the $\parent$ relation, but it may use formulas $\lcd_i$ for $i\in [0,d]$ as atoms.
That is, an lcd-reduced formula is a boolean combination of  label tests  and formulas $\lcd_i$.
Note that lcd-reduced formulas are closed under boolean combinations.

Suppose $\bar x=(x_1,\ldots,x_k)$ is a tuple of variables.
For each choice of functions $\gamma\colon \bar x\to \mathcal{P}(\Lambda)$ and $\delta\colon \bar x\times \bar x\to [0,d]$ we define the
{\em{lcd-type formula}} $$\basic_{\gamma,\delta}(\bar x)\in \FO[\{\parent\}\cup \Lambda]$$ as the lcd-reduced formula stating the following:
\begin{itemize}
\item for each $i\in [k]$, $c(x_i)$ holds for all $c\in \gamma(x_i)$ and $c(x_i)$ does not hold for all $c\in \Lambda-\gamma(x_i)$; and
\item for all $i,j\in [k]$, the number of common ancestors of $x_i$ and $x_j$ is $\delta(x_i,x_j)$.
\end{itemize}
Observe that the second condition implies that the depth of~$x_i$ is equal to $\delta(x_i,x_i)$.

Note that lcd-type formulas with $k$ free variables are mutually exclusive and cover all $k$-tuples: for each tuple of nodes $(u_1,\ldots,u_k)$ there is exactly one choice of $\gamma$ and $\delta$
for which $\basic_{\gamma,\delta}(u_1,\ldots,u_k)$ holds. Let $\Basic(k,d,\Lambda)$ be the set of lcd-type formulas with $k$ free variables for depth $d$ and label set~$\Lambda$; 
note that $\Basic(k,d,\Lambda)$ is finite and computable from $k,d,\Lambda$.

Observe that for an lcd-reduced formula $\varphi(\bar x)$ and a tuple~$\bar u$ of nodes with $|\bar u|=|\bar x|$, 
knowing which lcd-type formula from $\Basic(k,d,\Lambda)$ is satisfied on $\bar u$ is sufficient to determine whether $\varphi(\bar u)$ holds. This immediately yields the following.
\begin{proposition}\label{prop:bnf}
For every lcd-reduced formula $\varphi(\bar x)\in \FO[\{\parent\}\cup \Lambda]$ there exists a formula $\varphi'(\bar x)\in \FO[\{\parent\}\cup \Lambda]$ of the form
$$\varphi'(\bar x)=\bigvee_{\beta\in I} \beta(\bar x)$$
for some $I\subseteq \Basic(k,d,\Lambda)$
such that $\varphi$ and $\varphi'$ are equivalent over every $\Lambda$-labeled forest of depth at most $d$. Furthermore, $\varphi'$ can be computed from $d$, $\Lambda$, and $\varphi$.
\end{proposition}

An lcd-reduced formula $\varphi(\bar x)$ is in the {\em{basic normal form}} if it is a disjunction of basic formulas applied to $x_1,\ldots,x_k$, that is, it is of the form
$$\varphi(\bar x)=\bigvee_{\beta\in I} \beta(\bar x)$$
for some $I\subseteq \Basic(k,d,\Lambda)$.

\paragraph*{Quantifier elimination.}
We now show how to eliminate  a single existential quantifier for bounded depth forests. This will be used later for quantifier elimination on low-treedepth decompositions.

\begin{lemma}\label{lem:qe-trees}
Let $d\in \N$ and $\Lambda$ be a label set.
Then for every formula $\varphi(\bar x)\in \FO[\{\parent\}\cup \Lambda]$ with $|\bar x|\geq 1$ and of the form
$$\varphi(\bar x)=\exists_y\ \psi(\bar x,y)$$
where $\psi$ is lcd-reduced, and every $\Lambda$-labeled forest $\Ts$ of depth at most $d$, there exists a label set $\wLambda$,
an lcd-reduced formula $\wphi(\bar x)\in \FO[\{\parent\}\cup \wLambda]$, and a $\wLambda$-relabeling $\Ss$ of $\Ts$ such that $\varphi(\Ts)=\wphi(\Ss)$.

Moreover, the following effectiveness assertions hold. The label set $\wLambda$ is computable from~$d$ and~$\Lambda$, the formula $\wphi$ is computable from $\varphi,d,\Lambda$, and
the following transformation which computes~$\Ss$ given $\Ts$, parameterized by $\varphi,d,\Lambda$ is in $\linFPT$ and in $\paraAC{0\uparrow}$.
\end{lemma}
\begin{proof}
We first present the combinatorial construction of $\wLambda$, $\Ss$, and $\wphi$.
The effectiveness assertions will be discussed at the end.

We start by observing that without loss of generality we may assume that $\psi(\bar x,y)$ is an lcd-type formula.
Indeed, by Proposition~\ref{prop:bnf} we may assume that $\psi(\bar x,y)$ is of the form $\bigvee_{\beta\in I} \beta(\bar x,y)$
for some $I\subseteq \Basic(k+1,d,\Lambda)$. Since existential quantification commutes with disjunction, $\varphi$ is equivalent to the formula 
$\bigvee_{\beta\in I} (\exists_y\ \beta(\bar x,y))$.
Suppose now that for each $\beta\in I$ we find a suitable lcd-reduced formula $\widehat{\beta}(\bar x)$ and a $\wLambda_\beta$-recoloring $s_\beta$ of $\Ts$.
Having set $\wLambda\coloneqq \bigcup_{\beta\in I} \wLambda_\beta$, 
we can take~$\Ss$ to be the union relabeling, obtained by including all labels from all relabelings $(s_\beta)_{\beta\in I}$,
and define $\wphi(\bar x)\coloneqq \bigvee_{\beta\in I} \widehat{\beta}(\bar x)$.
Note that thus $\wphi$ is lcd-reduced.

Therefore, from now on we assume that 
$$\psi(\bar x,y)=\basic_{\gamma,\delta}(\bar x,y)$$ for some $\gamma\colon \bar z\to \mathcal{P}(\Lambda)$ and $\delta\colon \bar z\times \bar z\to [0,d]$, where $\bar z$ is $\bar x$ with $y$ appended.

Let us inspect values $\delta(x_i,y)$ for $i\in [k]$ and without loss of generality assume that $\delta(x_1,y)$ is the largest of them; here we use the assumption that $k\geq 1$.
Denote $h_y=\delta(y,y)$, $h_1=\delta(x_1,x_1)$, and $h=\delta(x_1,y)$.
Note that we may assume that $h_y$ and $h_1$ are larger or equal to $\max(h,1)$, for otherwise $\psi$ is not satisfiable and we may return an always false formula as $\wphi$.

From now on we assume for simplicity that $h>0$, that is, $x_1$ and $y$ have a common ancestor. The case $h=0$ is essentially the same and differs in notation details; we discuss it at the end.

Call a node $w$ in a $\Lambda$-labeled forest $\Ts$ a {\em{candidate}} if $w$ has depth $h_y$ and {\em{label pattern}} $\gamma(y)$:
$c(w)$ holds for all $c\in \gamma(y)$ and $c(w)$ does not hold for all $c\in \Lambda-\gamma(y)$.
For a node $v$ of $\Ts$ at depth $h$, we define an integer $\kappa(v)$ as follows:
$$\kappa(v)\coloneqq \textrm{the number of subtrees rooted at children of $v$ that contain a candidate.}$$
Note that in case $h=h_y$ the above value is meaningless --- it is always equal to $0$. 
Hence, in this case we redefine $\kappa(v)$ to be equal to $1$ if $v$ is a candidate and to $0$ if $v$ is not a candidate.

We now define the relabeling $\Ss$ of $\Ts$.
First, we include in $\Ss$ all the labels from $\Ts$; thus the final label set $\wLambda$ will be a superset of $\Lambda$.
Next, for every node $u$ record the following finite information using new labels at $u$:
\begin{enumerate}
\item[(a)] \label{e:witnesses} Provided $u$ is at depth $h_1$, record $\kappa(v)$ where $v$ is the unique ancestor of $u$ at depth $h$, or $\infty$ if this number is larger than $k$.
\item[(b)] \label{e:ancestors} Provided $u$ is at depth larger than $h$, record whether there exists a candidate $w$ such that the lowest common ancestor of $u$ and $w$ is at depth larger than $h$.
\end{enumerate}
The above information can be recorded using $k+2$ new labels: $k+1$ to record possible values in~Item~\ref{e:witnesses}, and $1$ to record the boolean value in Item~\ref{e:ancestors}.
Thus, $\wLambda$ is obtained by adding these $k+2$ new labels to $\Lambda$.

Having defined $\Ss$, we now write an lcd-reduced formula $\wphi(\bar x)\in \FO[\{\parent\}\cup \wLambda]$ that, when applied on some evaluation of $\bar x$ in $\Ss$,
verifies whether $\exists_y\ \psi(\bar x,y)$ holds in~$\Ts$.
Obviously we may start with verifying that:
\begin{itemize}
\item the label pattern of $x_i$ is equal to $\gamma(x_i)$, for each $i\in [k]$; and
\item for all $1\leq i\leq j\leq k$, the number of common ancestors of $x_i$ and $x_j$ is $\delta(x_i,x_j)$.
\end{itemize}
Both these checks can be easily done by lcd-reduced formulas. 
We are left with writing an lcd-reduced formula which verifies that, in addition to the above, a suitable node $y$ exists.
Actually, this formula will only use the new labels at nodes $x_1,\ldots,x_k$.

Observe that if nodes $u_1,\ldots,u_k,w$ in a forest $\Ts$ are such that $\Ts\models \psi(u_1,\ldots,u_k,w)$, then the subforest of $\Ts$ (with labels forgotten) 
formed the ancestors of $u_1,\ldots,u_k,w$ is uniquely determined by the function $\delta$, up to isomorphism.
Let this unique forest be $f$; note that $f$ can be either uniquely constructed from $\delta$ alone, or values contained in $\delta$ witness that no such $f$ exists 
and $\psi$ is not satisfiable, in which case we may return an always false formula as $\wphi$.
We may naturally label some nodes of $f$ with variables $x_1,\ldots,x_k,y$; note that thus each leaf of $f$ is labelled.
Whenever nodes $u_1,\ldots,u_k,w$ in a tree $\Ts$ are such that $\Ts\models \psi(u_1,\ldots,u_k,w)$, then we may define a natural isomorphism 
$\eta$ from $f$ to the ancestor closure of $\{u_1,\ldots,u_k,w\}$ in $\Ts$ by first mapping $y$ to $w$ and
each $x_i$ to $u_i$, and then extending the mapping to their ancestors.

We first resolve the corner case when $h=h_y$; equivalently $y$ is an ancestor of $x_1$ in $f$.
We claim that then it suffices to check whether the information encoded at $x_1$ in point (a) of the construction of
$f$ asserts that in the unique ancestor of $x_1$ at depth $h=h_y$ the value of $\kappa(\cdot)$ is positive.
Indeed, this unique ancestor needs to be the evaluation of $y$, and since $h=h_y$, we have that $\kappa(\cdot)$ at this node is $0$ 
if it does not have label pattern~$\gamma(y)$, and $1$ if it has label pattern~$\gamma(y)$ (recall that in case $h=h_y$ we redefined 
$\kappa(v)$ to be equal to $1$ if $v$ is a candidate and to 
$0$ if $v$ is not a candidate).
Hence, from now on we assume that $h_y>h$.

Let $z$ be the unique ancestor at depth $h$ of $y$ in $f$; equivalently $z$ is the least common ancestor of $y$ and $x_1$ in $f$.
Note that $z$ exists by the assumption $h>0$ and $z\neq y$ by the assumption $h_y>h$.
Let $q$ be the child of $z$ that is also an ancestor of $y$. Then the subtree of $f$ rooted at $q$ is a path with $y$ as the only leaf, and this subtree does not contain any node from $\{x_1,\ldots,x_k\}$.
Let $p_1,\ldots,p_\ell$ be the other children $z$ in $f$, i.e.\ those that are not ancestors of $y$. Note that $\ell\leq k$.

\begin{figure}[htbp!]
                \centering
		\def\svgwidth{0.5\textwidth}
                \begin{footnotesize}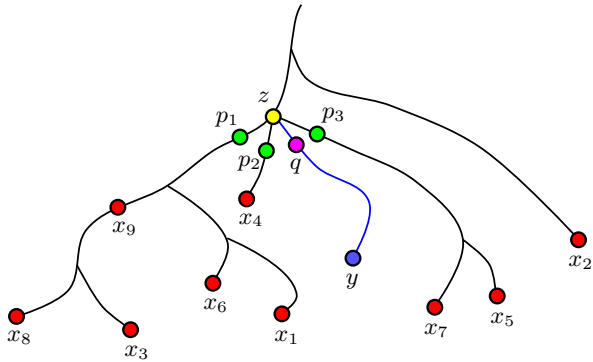\end{footnotesize}
        \caption{Example forest $f$ and nodes mentioned in the proof. In this example, to compute the value of $\ell'$ it suffices to count for how many of (the evaluations of) variables $x_1,x_4,x_5$ we have recorded
        information ``true'' in point Item~\ref{e:ancestors} of the construction of $\Ss$.}\label{fig:gadgets}
\end{figure}

Suppose we have a $\Lambda$-labeled tree $\Ts$ and nodes $u_1,\ldots,u_k,w$ of $\Ts$ such that $\Ts\models \psi(u_1,\ldots,u_k,w)$. Let $\eta$ be the natural isomorphism from $f$ to the ancestor-closure 
of $\{u_1,\ldots,u_k,w\}$ in $\Ts$.
Observe that knowing the labels at $u_1,\ldots,u_k$ in $\Ss$, we may uniquely deduce which of the subtrees of $\Ts$ rooted at the children $\eta(p_1),\ldots,\eta(p_\ell)$
of $\eta(z)$ contain a candidate. Indeed, each subtree of $\Ss$ rooted at some $p_j$, $j\in [\ell]$, contains some node $x_i$, $i\in [k]$, 
and then the corresponding subtree of~$\Ts$ rooted at $\eta(p_j)$ contains a candidate if and only if we 
recorded information ``true'' at the node $u_i$ in point (b) of the construction of $\Ss$.
Let $\ell'\leq \ell$ be the number of those subtrees.

Now comes the crux: having evaluated $x_1,\ldots,x_k$ to $u_1,\ldots,u_k$,
in order to verify whether a suitable evaluation of $y$ exists it suffices to check whether $\kappa(v)>\ell'$, where $v$ is the unique ancestor of $u_1$ at depth $h$; 
note that $v=\eta(z)$ in the notation of the previous paragraph.
Indeed, if such evaluation $w$ of $y$ exists, then the subtree rooted at a child of $v$ containing $w$ is not among the ones contributing to the value of $\ell'$, 
and hence it witnesses that $\kappa(v)>\ell'$. 
Conversely, if $\kappa(v)>\ell'$ then there exists a child $a$ of $v$ different from $\eta(p_1),\ldots,\eta(p_\ell)$, 
such that $a$ has a candidate $w$ as a descendant. Then evaluating $y$ to $w$ makes $\psi(u_1,\ldots,u_k,w)$ satisfied.

Note that since $\ell'\leq \ell\leq k$, the inequality $\kappa(v)>\ell'$ will always hold if $\kappa(v)>k$.
Therefore, the information recorded at $u_1$ in point (a) of the construction of $\Ss$ is sufficient to determine whether $\kappa(v)>\ell'$.
The computation presented above --- determination of $\ell'$ and verification whether $\kappa(v)>\ell'$ --- can be easily encoded using a formula that accesses only 
labels at nodes $u_1,\ldots,u_k$ in $\Ss$.
This concludes the construction of $\wphi$ in the case $h>0$. 

When $h=0$ the difference is that the node $z$ --- the lowest common ancestor of $y$ and $x_1$ in~$f$ --- does not exist.
However, then $y$ is in a different tree of the forest $f$ than all other nodes $x_1,\ldots,x_k$. We may define $q$ to be the root of the tree of $f$ containing $y$ and $p_1,\ldots,p_\ell$ to be the
roots of all other trees of $f$. Also, in point (a) of the construction of $\Ss$ we replace $\kappa(v)$ (as now $v$ does not exist) with the number of trees in $f$ that contain a candidate.
The rest of the reasoning is exactly the same.

\smallskip 

We are left with discussing the effectiveness assertions.
We discuss the case when $\psi$ is an lcd-type formula, as the lift to arbitrary lcd-reduced formulas is immediate, both for $\linFPT$ and $\paraAC{0}$.
It is clear that $\wLambda$ is computable from $d,\Lambda$ and $\wphi$ is computable from $\varphi,d,\Lambda$.

First, we argue that $\Ss$ can be constructed in linear FPT time.
This boils down to computing information added in points (a) and (b) for each node $u$ of $\Ts$.
For point (a), it suffices to run a suitable depth-first search in $\Ts$, where the return value from a subtree is the information whether it contains a candidate.
For point (b), we may first use depth-first search to mark all nodes of $\Ts$ that have a descendant that is a candidate.
Then we may apply a second depth-first search that computes the sought information; here it suffices to remember whether on the path from a root to the current node there was a marked node at depth larger than $h$.

Finally, to implement the transformation from $\Ts$ to $\Ss$ in $\paraAC{0\uparrow}$ we do the following.
First, compute the ancestor relation in $\Ts$ by computing the $d$th boolean power of the matrix of the parent relation in $\Ts$; this can be done by a polynomial-size circuit of depth $\Oh(\log d)$ by iterative squaring.
Next, for every pair of nodes $x$ and $y$ compute the number of common ancestors of $x$ and $y$; this can be done in $\paraAC{0}$ using Theorem~\ref{thm:threshold}.
Using the above,
the information recorded in point~\ref{e:witnesses} for each node $u$ can be computed in $\paraAC{0}$ using Theorem~\ref{thm:threshold},
while the information recorded in point Item~\ref{e:ancestors} can be computed in $\paraAC{0}$ directly from the definition.
\end{proof}

\subsection{Quantifier elimination for bounded expansion classes}

\paragraph*{Skeletons.} We first introduce {\em{skeletons}}, which are relational structures formed by putting a bounded number of forests of bounded depth on top of each other.
Essentially, they will be our abstraction for low treedepth decompositions.

Let $\Gamma$ be a vocabulary (of arity $2$) and let $d\in \N$. 
A $\Gamma$-structure $\As$ is a {\em{$\Gamma$-skeleton of depth $d$}} if 
for every binary relation $R\in \Gamma$, the structure  obtained from $\As$ by dropping all relations apart from $R$ and preserving the universe is a rooted forest of depth at most $d$, with $R$ serving the role
of the parent relation. Note that thus {\em{all}} binary relations in $\Gamma$ serve the roles of bounded-depth forests. 
The reader should think of a skeleton of depth $d$ as the union of several depth-$d$ forests
over the same universe, which is moreover labeled with some label set (i.e.\ with unary relations from $\Gamma$).

For every binary relation $R\in \Gamma$ and $i\in [0,d]$ we may construct a formula $\lcd_i^R(x,y)$ as in Proposition~\ref{prop:lcd}, but using~$R$ instead of $\parent$.
As before, a formula $\varphi(\bar x)\in \FO[\Gamma]$ is {\em{lcd-reduced}} if it does not use any quantifiers or binary relations, 
but may use formulas $\lcd_i^R$ as atoms; thus, it is a boolean combination of label checks and atoms $\lcd_i^R$.
We note that lcd-reduced formulas can be easily turned into existential formulas.

\begin{lemma}\label{lem:bnf-existential}
For every $d\in \N$ and lcd-reduced formula $\varphi(\bar x)\in \FO[\Gamma]$ there exists an existential formula $\xi(\bar x)\in \FO[\Gamma]$ 
such that $\varphi$ and $\xi$ are equivalent over every $\Gamma$-skeleton of
depth at most~$d$. Furthermore, $\xi$ can be computed from $d$, $\Gamma$, and $\varphi$.
\end{lemma}
\begin{proof}
Since $\varphi(\bar x)\in \FO[\Gamma]$ is a boolean combination of label checks and atoms $\lcd_i^R$, it may be rewritten as an equivalent formula $\zeta(\bar x)$ in conjunctive normal form:
$\zeta(\bar x)$ is a conjunction over clauses, where each clause is a disjunction of literals: label checks, atoms $\lcd_i^R$, and their negations. Observe that
each negation of an atom $\lcd_i^R(x,y)$ may be replaced by the formula 
$$\bigvee_{j\in [0,d],\, j\neq i}\ \lcd_j^R(x,y).$$
By applying such replacement exhaustively in $\zeta(\bar x)$ we obtain a formula $\zeta'(\bar x)$ that is equivalent to $\zeta(\bar x)$ over every $\Gamma$-skeleton of depth $d$.
Moreover, $\zeta(\bar x)$ is again in conjunctive normal form, however now atoms $\lcd_i^R$ appear only positively. 
By replacing these atoms with existential formulas provided by Proposition~\ref{prop:lcd} we turn $\zeta'(\bar x)$ into a positive boolean combination of existential formulas,
hence into an existential formula.
\end{proof}

For future reference we observe that lcd-reduced formulas can be efficiently evaluated. 
Let us remark that as far as sequential algorithms are concerned, we always assume that all rooted forests, including
forests in skeletons, are encoded by specifying for each element its parent; thus we assume that this parent can be found in constant time in the RAM model.

\begin{lemma}\label{lem:lcd-evaluate}
The following problem parameterized by \mbox{$d\in \N$} and an lcd-reduced formula $\alpha(\bar x)\in \FO[\Gamma]$ 
is in $\paraAC{0\uparrow}$ and can be computed in time $\Oh(d|\alpha|)$:
given a $\Gamma$-skeleton $\As$ of depth $d$ and a tuple $\bar u\in V(\As)^{|\bar x|}$, verify whether $\As\models \alpha(\bar u)$.
\end{lemma}
\begin{proof}
Observe that in $\alpha(\bar u)$ every atom of the form $\lcd^R_h(u_i,u_j)$ can be evaluated in time $\Oh(d)$ by pursuing the parent relation $R$ from $u_i$ and $u_j$ at most $d$ times.
Moreover, for each binary relation $R\in \Sigma$, interpreted as a parent relation in a forest of depth at most $d$, we can compute the corresponding ancestor relation using a circuit of size $n^{\Oh(1)}$
and depth $\Oh(\log d)$: just compute the $d$th boolean power of the matrix of $R$ using iterative squaring. With this information available, all atoms of the form $\lcd^R_h(u_i,u_j)$ can be evaluated in
$\paraAC{0}$ using Theorem~\ref{thm:threshold}. Having evaluated all the atoms it is straightforward to evaluate the whole formula within the stated complexity bounds.
\end{proof}

\paragraph*{Guarded structures.}
For the remainder of this section we fix a graph class $\Cc$ with effectively bounded expansion. Without loss of generality we may assume that $\Cc$ is closed under taking subgraphs.
Let us fix the function $M(\cdot)$ given by Theorem~\ref{thm:lowtd-colorings} for the class $\Cc$; this means that given a graph $G\in \Cc$ and parameter $p$ we may compute a treedepth-$p$ coloring of~$G$ using at most $M(p)$
colors in $\paraAC{1}$. 
We note that the same computational task can be also done in $\linFPT$ by the results of Ne\v{s}et\v{r}il and Ossona de Mendez~\cite{nevsetvril2008grad,NesetrilM08a}.

A structure $\As$ is {\em{guarded}} by $\Cc$ if the Gaifman graph of $\As$ belongs to $\Cc$.
Further, a structure~$\Bs$ with the same universe as $\As$ (but possibly different vocabulary) is {\em{guarded}} by $\As$ if the Gaifman graph of $\Bs$ is a subgraph of the Gaifman graph of~$\As$.
Note that if $\As$ guards $\Bs$ and $\Bs$ guards $\Cs$ then $\As$ guards~$\Cs$, and if further $\As$ is guarded by $\Cc$, then so are $\Bs$ and $\Cs$.

\paragraph*{Quantifier elimination.}
We finally proceed to our main goal, the quantifier elimination procedure for $\FO$ on structures with Gaifman graphs from $\Cc$.
The following definition explains our goal in this procedure.

\begin{definition}\label{def:reducible}
Let $\Sigma$ be a vocabulary (of arity $2$) and let $\varphi(\bar x)\in \FO[\Sigma]$ for a tuple of variables $\bar x=(x_1,\ldots,x_k)$.
We say that $\varphi(\bar x)$ is {\em{reducible}} if there exists $d\in \N$, a vocabulary~$\Gamma$, an lcd-reduced formula $\alpha(\bar x)\in \FO[\Gamma]$, and, for every $\Sigma$-structure~$\As$ guarded by $\Cc$,
a $\Gamma$-skeleton $\Bs$ of depth at most~$d$ guarded by $\As$ such that $\varphi(\As)=\alpha(\Bs)$.

This reducibility is {\em{effective}} if $d$, $\Gamma$, and $\alpha$ are computable from $\Sigma$ and $\varphi$, 
and the transformation computing $\Bs$ given $\As$, parameterized by $\Sigma$ and $\varphi$, is in $\linFPT$ and $\paraAC{1}$.
\end{definition}

Note that in Definition~\ref{def:reducible}, the fact that $\As$ is guarded by $\Cc$ and guards $\Bs$, entails that $\Bs$ is guarded by $\Cc$.
Hence we may further apply further reducibility on the structure $\Bs$ and so on. This chaining property of the notion of reducibility will be crucial in our reasoning.

In the following, whenever a vocabulary of a formula is not specified, it is an arbitrary vocabulary.
Given Definition~\ref{def:reducible}, quantifier elimination can be stated in a very simple way. 

\begin{theorem}\label{thm:qe-main}
Every formula $\varphi(\bar x)$ with $|\bar x|\geq 1$ is effectively reducible.
\end{theorem}

The proof of Theorem~\ref{thm:qe-main} is by induction on the structure of the formula.
We find it most convenient to directly solve the case of existential formulas first, from which both the induction base and the induction step will follow.

\begin{lemma}\label{lem:existential}
Every existential formula with at least one free variable is effectively reducible.
\end{lemma}
\begin{proof}
Let $\varphi(\bar x)\in \FO[\Sigma]$ be the formula in question, where~$\Sigma$ is a vocabulary and $\bar x=(x_1,\ldots,x_k)$ are the free variables of $\varphi$. 
We may assume that $\varphi$ is in prenex existential form, say,
$\varphi(\bar x)=\exists_{\bar y}\ \psi(\bar x,\bar y)$,
where $\bar y=(y_1,\ldots,y_\ell)$ and $\psi(\bar x,\bar y)$ is quantifier-free.
Suppose $\As$ is the given $\Sigma$-structure guarded by $\Cc$.
We describe how a suitable skeleton $\Bs$ guarded by $\As$ should be constructed, while its depth $d$, its vocabulary $\Gamma$, and the final lcd-reduced formula $\alpha(\bar x)$ will be constructed along the~way.

Let $p=k+\ell$ and let $G=G(\As)$.
Since $G\in \Cc$, there is a treedepth-$p$ coloring $\lambda\colon V(\As)\to [M]$ of $G$, where $M=M(p)$, which can be computed in $\paraAC{1}$ and in $\linFPT$ using Theorem~\ref{thm:lowtd-colorings}
and the results of~\cite{NesetrilM08a,nevsetvril2008grad}, respectively.
Let $\Classes$ be the family of all subsets of $[M]$ of size $p$.
For $C\in \Classes$, let $V^C=\lambda^{-1}(C)$ be the set of elements with colors from~$C$, and let $G^C=G[V^C]$ be the subgraph induced by them. Then~$G^C$ has treedepth at most $p$.
As observed before, we may compute, for each $C\in \Classes$, a DFS forest $F^C$ of the induced subgraph~$G^C$ of depth at most $d\coloneqq 2^p-1$ in $\paraAC{0\uparrow}$ (and in $\linFPT$ by just running a depth-first search).

Using Lemma~\ref{lem:DFS-compute}, we may compute, for each $C\in \Classes$, a DFS forest $F^C$ of the induced subgraph~$G^C$ of depth at most $d\coloneqq 2^p-1$. 
This can be done in $\paraAC{0\uparrow}$ by Lemma~\ref{lem:DFS-compute} and in $\linFPT$ by just running a depth-first search.

Fix any $C\in \Classes$ and let $\Ts^C$ be the unlabeled forest of depth at most $d$ with node set $V^C$ and binary relation $\parent^C$ interpreted as the parent relation of $F^C$.
We now prove that the substructure induced in $\As$ by $V^C$ can be entirely encoded in a labeling of $\Ts^C$.

\begin{claim}\label{cl:encode}
There exists a label set $\Lambda^C$, a $\Lambda^C$-labeling $\Ss^C$ of~$\Ts^C$, and, for every relation $R\in \Sigma$,
an lcd-reduced formula $\eta^C_R$ with as many free variables as the arity of $R$ such that $R(\As[V^C])=\eta^C_R(\Ss^C)$, where $\As[V^C]$ denotes the substructure of~$\As$ induced by $V^C$.
\end{claim}
\begin{proof}\renewcommand\proofSymbol{\ensuremath{\lrcorner}}
For each $u\in V^C$, record the following information using labels at $u$:
\begin{enumerate}
\item[(a)] \label{p:unary} For each unary relation $R\in \Sigma$, record whether $R(u)$ holds.
\item[(b)] \label{p:binary} For each binary relation $R\in \Sigma$ and each $i\in [d]$ not larger than the depth of $u$, let $v$ be the unique ancestor of $u$ at depth $i$.
Record whether $R(u,v)$ holds and whether $R(v,u)$ holds.
\end{enumerate}
Note that to record the information above we may use a set $\Lambda^C$ consisting of at most $2d|\Sigma|$ labels: one label for each unary relation in $\Sigma$,
and two labels for each binary relation in $\Sigma$ and each $i\in [d]$.
Let $\Ss^C$ be the obtained $\Lambda^C$-labeling of~$\Ts^C$.

We now write the lcd-reduced formula $\eta^C_R$ for a relation $R\in \Sigma$.
If $R(x)$ is unary, we may simply check an appropriate new label of $x$, added in point~\ref{p:unary}.
If $R(x,y)$ is binary, we make a disjunction over two cases: either $x$ is an ancestor of $y$ in $\Ss^C$ or vice versa.
If $x$ is an ancestor of~$y$, then denoting the depth of $x$ as $i$,
it can be checked whether $R(x,y)$ holds by reading one of the two labels added in point~\ref{p:binary} at $y$ for $R$ and depth $i$.
The case when $y$ is an ancestor of $x$ is symmetric.
If neither of $x$, $y$ is an ancestor of the other, then $R(x,y)$ is surely false, because $F^C$ is a DFS forest of the $G^C$.
This concludes the construction of $\eta^C_R$; note that the above verification can be indeed encoded using an lcd-reduced formula. 
\end{proof} 

Now consider formula
$\psi^C(\bar x,\bar y)\in \FO[\{\parent^C\}\cup \Lambda^C]$
obtained from $\psi(\bar x,\bar y)$ by replacing each relation symbol $R$ with the corresponding formula $\eta_R^C$.
Since $\psi$ was quantifier-free, $\psi^C$ is lcd-reduced.
Let
$$\varphi^C(\bar x)\coloneqq \exists_{\bar y}\ \psi^C(\bar x,\bar y).$$
Since $k\geq 1$, we may iteratively apply Lemma~\ref{lem:qe-trees} to consecutive quantifiers in $\varphi^C$, starting with the deepest.
This yields a new label set $\wLambda^C$, a $\wLambda^C$-relabeling $\Us^C$ of $\Ss^C$, and an lcd-reduced formula $\alpha^C(\bar x)$ such that $\alpha^C(\Us^C)=\varphi^C(\Ss^C)$.

We now build $\Bs$ and its vocabulary $\Gamma$. Start by setting the universe of $\Bs$ to be equal to the universe of $\As$, and $\Gamma$ is so far empty.
For each $C\in \Classes$ add a unary relation $\clr^C$ to $\Gamma$, and interpret it in $\Bs$ so that it selects the vertices of~$V^C$.
Next, import all the relations from all structures $\Us^C$ to $\Bs$.
That is, for each $C\in \Classes$ we add $\{\parent^C\}\cup \wLambda^C$ to the vocabulary $\Gamma$, while the interpretations of these relations are taken from~$\Ss^C$.
Note that thus elements outside of $V^C$ do not participate in relations $\parent^C$.

This concludes the construction of $\Gamma$ and $\Bs$. Note that the only binary relations in $\Gamma$ are the relations $\parent^C$ for $C\in \Classes$, and in $\Bs$ each of them
induces a forest of depth at most $d$.
Moreover, since each $F^C$ is a DFS forest of $G^C$, it follows that $\Bs$ is guarded by $\As$.
Hence, $\Bs$ is a $\Gamma$-skeleton of depth $d$ guarded by $\As$, as requested.
Consider the formula
$$\alpha(\bar x)\coloneqq \bigvee_{C\in \Classes}\ \left(\alpha^C(\bar x)\wedge \bigwedge_{i=1}^k\ \clr^C(x_i)\right).$$
Observe that $\alpha(\bar x)$, as a boolean combination of lcd-reduced formulas, is lcd-reduced.
We claim that $\alpha(\Bs)=\varphi(\As)$.
On one hand, by the construction it is clear that $\alpha$ selects only tuples that satisfy $\varphi$, thus $\alpha(\Bs)\subseteq \varphi(\As)$.
To see the reverse inclusion, observe that whenever we have some valuation $\bar u$ of $\bar x$ such that $\varphi(\bar u)$ holds, this is witnessed by the existence of some valuation $\bar v$ of $\bar y$
such that $\psi(\bar u,\bar v)$ holds. Since $|\bar u|+|\bar v|=k+\ell=p$ and $\lambda$ was a treedepth-$p$ coloring, there exists $C\in \Classes$ such that all elements of $\bar u$ and $\bar v$
belong to $V^C$. For this $C$ the formula $\alpha^C(\bar u)\wedge \bigwedge_{i=1}^k \clr^C(u_i)$ will be satisfied and, consequently, $\bar u$ will be included in $\alpha(\Bs)$.

This proves reducibility. Effectiveness follows immediately from complexity bounds provided by the invoked results and a straightforward implementation of the construction of Claim~\ref{cl:encode}. 
\end{proof}

We now use Lemma~\ref{lem:existential} to give all the ingredients needed for the inductive proof of Theorem~\ref{thm:qe-main}.

\begin{lemma}\label{lem:induction}
The following assertions hold:
\begin{enumerate}
\item[(a)] \label{a:qf} Every quantifier-free formula with at least one free variable is effectively reducible.
\item[(b)] \label{a:ng} The negation of an effectively reducible formula is effectively reducible.
\item[(c)] \label{a:eq} Every formula of the form $\varphi(\bar x)=\exists_{y}\,\psi(\bar x,y)$ for an effectively reducible $\psi$ and with $|\bar x|\geq 1$ is also effectively reducible.
\end{enumerate}
\end{lemma}
\begin{proof}
Assertion~\ref{a:qf} follows immediately from Lemma~\ref{lem:existential}, because every quantifier-free formula is in particular existential.
For assertion~\ref{a:ng}, after performing the constructions given by the effective reducibility of the formula $\varphi(\bar x)$ in question,
we may just negate the output lcd-reduced formula. Here we use that lcd-reduced formulas are closed under negation.

For assertion~\ref{a:eq}, apply first the effective reducibility of $\psi(\bar x,y)$, yielding an lcd-reduced formula $\beta(\bar x,y)$ working over some $\Gamma$-skeleton $\Cs$ guarded by $\As$.
Then apply Lemma~\ref{lem:bnf-existential} and Proposition~\ref{prop:prenex} to rewrite $\beta(\bar x,y)$ as an equivalent formula $\gamma(\bar x,y)$ in the prenex existential form.
Finally, observe that formula $\exists_y\ \gamma(\bar x,y)$ is also in the prenex existential form, so we may apply Lemma~\ref{lem:existential} to it and $\Cs$, 
yielding the final lcd-reduced formula $\alpha(\bar x)$ and skeleton $\Bs$.
\end{proof}

We may now conclude the proof of Theorem~\ref{thm:qe-main}.

\begin{proof}[of Theorem~\ref{thm:qe-main}]
It is well-known that every formula can be rewritten into prenex normal form, i.e., to the form where there is a sequence of quantifiers (existential or universal) followed by a quantifier-free formula.
Further, such a formula in the prenex normal form can be constructed from the quantifier-free formula using a sequence of negations and existential quantifications.
Reducibility now follows from Lemma~\ref{lem:induction}: assertion \ref{a:qf} states that the quantifier-free formula is reducible,
while assertions \ref{a:ng} and \ref{a:eq} imply that during this construction procedure we encounter only reducible formulas.

For effectiveness, a sequential composition of the constructions yields both computability of the final depth, vocabulary, and formula from the input vocabulary and formula, as well as containment in $\linFPT$
of the transformation.
For containment in $\paraAC{1}$, there is a slight caveat: a sequential composition of constructions would yield only containment in $\paraAC{1\uparrow}$, and not in $\paraAC{1}$.
We may solve this issue as follows. Observe that the only place in the proof where we need to use the computational power of $\paraAC{1}$ is when we invoke Theorem~\ref{thm:lowtd-colorings}
to compute a treedepth-$p$ coloring of the Gaifman graph; this happens in the beginning of the proof of Lemma~\ref{lem:existential}.
Since the final depth $d_{\mathrm{max}}$ is computable from the input vocabulary and formula, we may compute it beforehand. It is easy to see that during the construction procedure
we never use treedepth-$p$ colorings for any $p>d_{\mathrm{max}}$. Therefore, we may in the very beginning apply Lemma~\ref{lem:existential} on the Gaifman graph of the input structure 
for every value of $p\leq d_{\mathrm{max}}$ in {\em{parallel}}; this can be done in $\paraAC{1}$. Then all further applications of Lemma~\ref{lem:existential} can be replaced with the usage of a pre-computed
treedepth-$p$ coloring for an appropriate $p$. As the remainder of the construction actually works in $\paraAC{0\uparrow}$, the result follows.
\end{proof}

\subsection{Piecing together the proof of Theorem~\ref{thm:main}}

With quantifier elimination in place, we may conclude the proof of our main result, Theorem~\ref{thm:main}.

\begin{proof}[of Theorem~\ref{thm:main}]
Let $\As$ be an input $\Sigma$-structure on $n$ elements, for a vocabulary $\Sigma$ of arity at most $2$, and let $\varphi\in \FO[\Sigma]$ be the input sentence.
We would like to apply Theorem~\ref{thm:qe-main} to~$\varphi$. However there is a slight mismatch: Theorem~\ref{thm:qe-main} assumes that the input formula has at least one free variable.
To circumvent this, let $z$ be a fresh variable that is not used in $\varphi$ and let us consider $\varphi$ as a formula $\varphi(z)$ with one free variable $z$ that is never used.
Note that $\varphi(z)$ is true either for every element of $\As$ or for no element of $\As$, depending on whether $\varphi$ is true or false in $\As$.
Now apply Theorem~\ref{thm:qe-main} to $\varphi(z)$, yielding a $\Gamma$-skeleton $\Bs$ of some depth $d$ guarded by $\As$ and an lcd-reduced formula $\alpha(z)\in \FO[\Gamma]$ such that 
$\alpha(\Bs)=\varphi(\As)$. It remains to evaluate $\alpha$ on any element of the structure using Lemma~\ref{lem:lcd-evaluate}.
\end{proof}

The same reasoning allows to reprove the result of Dvo\v{r}\'ak et al.~\cite{dvovrak2013testing} that for every class of effectively 
bounded expansion $\Cc$, it can be verified in linear-FPT time whether an input sentence~$\varphi$ holds in a given structure whose Gaifman graph belongs to $\Cc$. 

A cautious reader might be a bit worried by the strange-looking introduction of the dummy variable $z$ in the proof above. Let us explain its combinatorial meaning.
Unraveling the proof of Theorem~\ref{thm:qe-main}, the need for this workaround is the assumption $k\geq 1$ in Lemma~\ref{lem:qe-trees}.
This assumption is essential for the proof of Lemma~\ref{lem:qe-trees} to work: the information about the existence of a suitable evaluation of $y$ is encoded in the new label of $x_1$, and we need to have this
variable $x_1$ in order to store the information somewhere. Moreover, the assumption is actually necessary for Lemma~\ref{lem:qe-trees} to hold as stated, as
an lcd-reduced formula with no free variables can only say ``true'' or ``false'', while whether a suitable evaluation of $y$ exists depends on the forest $\Ts$, and not just on the input formula $\varphi$.
A way to overcome the issue would be to generalize the notion of a labeled forest by allowing additional arity-$0$ relations (aka {\em{flags}}); then 
we could store the information in a flag in $\Ss$ and the output lcd-reduced sentence $\wphi$ would just output this flag.
The implemented resolution by adding a dummy free variable $z$ is a variant of this: the variable $z$ serves as a ``placeholder'' for storing the relevant information, which boils down to encoding an arity-$0$
relation as an arity-$1$ relation that is satisfied either in all or in no element of the structure.

%% file: qe-forest.pdf_tex
\begingroup%
  \makeatletter%
  \providecommand\color[2][]{%
    \errmessage{(Inkscape) Color is used for the text in Inkscape, but the package 'color.sty' is not loaded}%
    \renewcommand\color[2][]{}%
  }%
  \providecommand\transparent[1]{%
    \errmessage{(Inkscape) Transparency is used (non-zero) for the text in Inkscape, but the package 'transparent.sty' is not loaded}%
    \renewcommand\transparent[1]{}%
  }%
  \providecommand\rotatebox[2]{#2}%
  \ifx\svgwidth\undefined%
    \setlength{\unitlength}{271.25817522bp}%
    \ifx\svgscale\undefined%
      \relax%
    \else%
      \setlength{\unitlength}{\unitlength * \real{\svgscale}}%
    \fi%
  \else%
    \setlength{\unitlength}{\svgwidth}%
  \fi%
  \global\let\svgwidth\undefined%
  \global\let\svgscale\undefined%
  \makeatother%
  \begin{picture}(1,0.59084958)%
    \put(0,0){\includegraphics[width=\unitlength]{qe-forest.pdf}}%
    \put(0.44283262,0.03065711){\color[rgb]{0,0,0}\makebox(0,0)[lb]{\smash{$x_1$}}}%
    \put(0.41469862,0.42744841){\color[rgb]{0,0,0}\makebox(0,0)[lb]{\smash{$z$}}}%
    \put(0.56127596,0.12526219){\color[rgb]{0,0,0}\makebox(0,0)[lb]{\smash{$y$}}}%
    \put(0.34518629,0.395115){\color[rgb]{0,0,0}\makebox(0,0)[lb]{\smash{$p_1$}}}%
    \put(0.52184486,0.40126239){\color[rgb]{0,0,0}\makebox(0,0)[lb]{\smash{$p_3$}}}%
    \put(0.38218663,0.32125004){\color[rgb]{0,0,0}\makebox(0,0)[lb]{\smash{$p_2$}}}%
    \put(0.46801822,0.31688561){\color[rgb]{0,0,0}\makebox(0,0)[lb]{\smash{$q$}}}%
    \put(-0.00136625,0.03029557){\color[rgb]{0,0,0}\makebox(0,0)[lb]{\smash{$x_8$}}}%
    \put(0.19357321,0.00556435){\color[rgb]{0,0,0}\makebox(0,0)[lb]{\smash{$x_3$}}}%
    \put(0.3832358,0.22623346){\color[rgb]{0,0,0}\makebox(0,0)[lb]{\smash{$x_4$}}}%
    \put(0.32450269,0.08412207){\color[rgb]{0,0,0}\makebox(0,0)[lb]{\smash{$x_6$}}}%
    \put(0.80166802,0.06084566){\color[rgb]{0,0,0}\makebox(0,0)[lb]{\smash{$x_5$}}}%
    \put(0.69110531,0.04193368){\color[rgb]{0,0,0}\makebox(0,0)[lb]{\smash{$x_7$}}}%
    \put(0.93114263,0.15540592){\color[rgb]{0,0,0}\makebox(0,0)[lb]{\smash{$x_2$}}}%
    \put(0.1732064,0.21214204){\color[rgb]{0,0,0}\makebox(0,0)[lb]{\smash{$x_9$}}}%
  \end{picture}%
\endgroup%

%% file: conclusions.tex
\section{Conclusions}\label{sec:conclusions}

In this paper we showed that the model-checking problem for first-order logic on classes of effectively bounded expansion is in $\paraAC{1}$, 
which means that it can be solved by a family of \AC-circuits of size $f(\varphi)\cdot n^{\Oh(1)}$ and
depth $f(\varphi)+\Oh(\log n)$, where $f$ is a computable function. This can be regarded as a parallelized variant of the result of Dvo\v{r}\'ak et al.~\cite{dvovrak2013testing} stating that the problem is
fixed-parameter tractable.

By the result of Grohe et al.~\cite{grohe2014deciding}, model-checking $\FO$ is fixed-parameter tractable even on every nowhere dense class of structures. When trying to generalize our result to the nowhere dense 
setting, the main issue is that the proof of Grohe et al.~\cite{grohe2014deciding} does not yield a robust quantifier elimination procedure, but a weak variant of Gaifman local form that is sufficient for 
fixed-parameter tractability of model-checking, but not variations of the problem. 

Our techniques uncover tight connections between the paradigms of distributed computing and circuit complexity in the context of sparse graphs classes.
Methods of the theory of sparsity seem very well-suited for the design of distributed algorithms, yet so far little is known.
Ne\v{s}et\v{r}il and Ossona de Mendez gave a logarithmic-time distributed algorithm to compute low treedepth colorings on classes of bounded expansion~\cite{NesetrilM16}. 
In the light of this paper, it is very natural to repeat the question asked by Ne\v{s}et\v{r}il and Ossona de Mendez~\cite{NesetrilM16} of whether on every class of bounded expansion,
model-checking local first-order formulas
can be performed by a distributed algorithm with running time $f(\varphi)\cdot \log n$ in the {\em{local broadcast}} model. As computation of low treedepth colorings is already settled~\cite{NesetrilM16}, it remains to examine the
quantifier elimination procedure; we hope that our presentation of this argument may help with this. Stronger models of communication (such as the so-called {\em{congested clique}} model) may allow efficient
distributed algorithms for more general problems, like model-checking of first-order formulas that are not necessarily local.


%% file: app-degeneracy.tex
\section{Hardness of computing degeneracy exactly}\label{app:degeneracy}

In this section we prove \cref{thm:degeneracy-Phard}, that is, we show that the problem of determining whether an input graph has degeneracy at most $2$ is $\mathsf{P}$-hard under logspace reductions.
We reduce from the following {\sc{Circuit Evaluation}} problem.
We are given an \AC-circuit $C$, with one output gate and no restrictions on depth, and an evaluation of the input gates of $C$.
The task is to determine whether the only output gate of $C$ evaluates to $1$.
It is well-known that this problem is $\mathsf{P}$-complete under logspace reductions, 
since the computation of a polynomial-time deterministic Turing machine can be encoded as a polynomial-size \AC-circuit.

\paragraph*{Intuition.}
We first present the intuition behind the reduction.
Suppose we are given a graph $G$ and we are interested in finding out whether its degeneracy is at most $2$.
Consider the following {\em{elimination procedure}}: starting with the original graph $G$, iteratively remove a vertex of degree at most $2$ from the current graph as long as there is one.
If we manage to exhaust the whole vertex set of $G$ by the elimination procedure, then by ordering the vertices according to the time of their removal we obtain an ordering of degeneracy
at most $2$. Otherwise, if the elimination procedure gets stuck at a subgraph with minimum degree at least $3$, then by \cref{prop:subgraph} this witnesses that the degeneracy of the graph is at least $3$.

The idea for the reduction is as follows: given the input circuit $C$ we construct a graph $G$ by replacing each gate by an appropriate gadget so that the elimination procedure on $G$ 
corresponds to the natural bottom-up evaluation procedure for the gates of $C$. More precisely, a gadget gets removed in the elimination procedure if and only if the corresponding gate evaluates to~$1$. 
If the output gate of $C$ evaluates to $1$ --- which means that the corresponding gadgets gets removed --- then this triggers a special ``switch'' that makes all gadgets removable, and hence~$G$ has degeneracy
at most $2$. Otherwise, all gadgets corresponding
to gates that evaluate to $0$ induce a subgraph of minimum degree at least $3$, thus certifying that $G$ has degeneracy at least $3$.

\paragraph*{Preprocessing.}
We first make some preprocessing of the input circuit $C$ in order to streamline the construction; it will be straightforward to see that this pre-processing can be done in logarithmic space.
Suppose the input gates of $C$ are $x_1,\ldots,x_n$, the output gate of $C$ is $y$, and we are also given an evaluation $\eta\colon \{x_1,\ldots,x_n\}\to \{0,1\}$
of the input gates. We now introduce $\TRUE$ and $\FALSE$ gates; such gates have always fan-in $0$ and evaluate to $1$ and $0$, respectivley.
For each $x_i$ such that $\eta(x_i)=1$, replace the input gate $x_i$ with a $\TRUE$ gate.
Similarly, if $\eta(x_i)=0$ then replace $x_i$ with a $\FALSE$ gate.
Moreover, every $\OR$ gate with no input is replaced with a $\FALSE$ gate and every $\AND$ gate with no input is replaced by a $\TRUE$ gate.
Further, by the standard technique of eliminating negation using de Morgan's laws we may assume that $C$ also has no $\NOT$ gates.
Thus, from now on we may assume that $C$ has only $\AND$, $\OR$, $\TRUE$, and $\FALSE$ gates, and every $\AND$ and $\OR$ gate has fan-in at least $1$.

Next, since we are not concerned with the depth of the circuit, 
we may assume that every gate has fan-in at most $2$ by replacing each gate of larger fan-in with a binary tree of gates of the same type of fan-in $2$.
Similarly, we may further assume that every gate has fan-out at most~$2$, i.e. it is wired as the input to at most $2$ other gates. 
For this, we modify every gate $u$ with fan-out $d>2$ by adding a path consisting of $d-2$ $\AND$-gates $u_1,u_2,\ldots,u_{d-2}$ with fan-in $1$ and fan-out $2$ arranged as follows:
denoting $u_0=u$, every gate $u_i$ for $i\geq 1$ takes as the only input the gate $u_{i-1}$.
Thus, the additional $\AND$-gates copy the evaluation of $u$ and in total all the gates $u,u_1,\ldots,u_{d-2}$ have fan-out $d$; this fan-out $d$ can be used to re-wire the $d$ wires originally going out of $u$.

Finally, every $\OR$-gate with fan-in $1$ is replaced by an $\AND$-gate with the same input and output --- they have exactly the same functionality.
All in all, we have achieved the following properties:
\begin{itemize}
\item circuit $C$ has only $\AND$, $\OR$, $\TRUE$, and $\FALSE$ gates, out of which there is one output gate $y$;
\item each $\AND$ gate of $C$ has fan-out at most $2$ and fan-in $1$ or $2$;
\item each $\OR$ gate of $C$ has fan-out at most $2$ and fan-in $2$.
\end{itemize}
The problem is to determine whether the output gate $y$ evaluates to $1$.

\paragraph*{Gadgets.}
We now present the gadgets for the $\OR$/$\AND$ gates; they are depicted in \cref{fig:gadgets}.

\begin{figure}[htbp!]
                \centering
		\def\svgwidth{0.8\textwidth}
                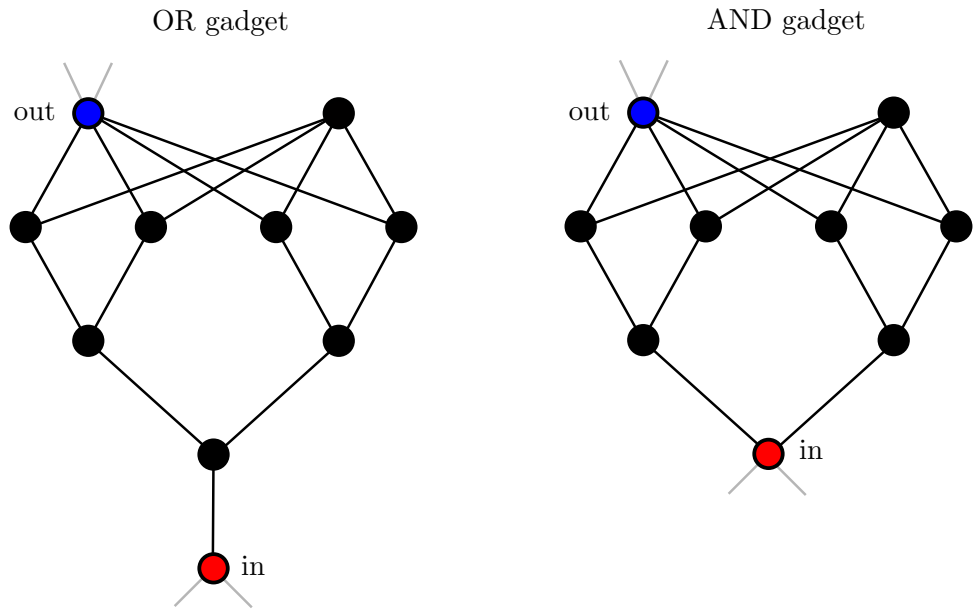
        \caption{Gadgets for the $\OR$/$\AND$ gates. Input vertices are depicted in red, output vertices are depicted in blue. The grey edges outgoing from the input and output vertices depict places where connections
        to other gadgets are attached.}\label{fig:gadgets}
\end{figure}

The $\OR$ gadget has $10$ vertices out of which there are two named: one {\em{input vertex}} named ``$\inpv$'', and one {\em{output vertex}} named ``$\outv$''.
Similarly for the $\AND$ gadget. We make the following two observations. First, in each gadget there is only one vertex of degree smaller than $3$, which is the input vertex; in the $\OR$ gadget it
has degree $1$, and in the $\AND$ gadget it has degree $2$. Further, each gadget has a vertex ordering with degeneracy $2$ --- just order the vertices in a top-down manner.

In the final construction we will use many copies of these gadget. For a copy $A$ of the $\OR$/$\AND$ gadget, by $\inpv[A]$ and $\outv[A]$ we denote the input and output vertiex in $A$, respectively,
and by $\sigma[A]$ we denote the abovementioned top-down vertex ordering of $A$ with degeneracy $2$.

\paragraph*{Construction.}
We now present the construction of the graph $G$ out of the circuit $C$. We first introduce gadgets to reflect the structure of $C$.
\begin{enumerate}
\item For each $\AND$ gate $x$ in $C$, introduce a copy $A_x$ of the $\AND$ gadget.
\item For each $\OR$ gate $x$ in $C$, introduce a copy $A_x$ of the $\OR$ gadget.
\item For each $\TRUE$ or $\FALSE$ gate $x$ in $C$, introduce a copy $A_x$ of the $\AND$ gadget.
\item Whenever there is a wire from gate $z$ to gate $x$ (i.e. $z$ is an input to $x$), add an edge between $\outv[A_z]$ and $\inpv[A_x]$.
\end{enumerate}
Next, we perform the following construction similar to the reduction of fan-out.
Let $k$ be the number of $\FALSE$ gates.
Introduce $k$ $\AND$ gadgets $B_1,B_2,\ldots,B_k$, and denoting $A_y=B_0$ (recall that $y$ is the output gate of $C$), add an edge between $\outv[B_{i-1}]$ and $\inpv[B_i]$ for each $i\in [k]$.
Finally, to each $\FALSE$ gate $x$ assign a different integer $i\in [k]$ and add an edge between $\outv[B_i]$ and $\inpv[A_x]$. This concludes the construction; it is clear that it can be implemented
in logspace.

\paragraph*{Correctness.}
To verify the correctness of the reduction we need to prove the following lemma.

\begin{lemma}\label{lem:correctness}
Gate $y$ evaluates to $1$ in $C$ if and only if the degeneracy of $G$ is at most $2$.
\end{lemma}

The proof of \cref{lem:correctness} is divided into two lemmas, each showing one implication.

\begin{lemma}\label{lem:eval-deg}
If gate $y$ evaluates to $1$ in $C$, then $G$ has degeneracy at most $2$.
\end{lemma}
\begin{proof}
Since $C$ is a circuit, it is acyclic and therefore there exists a topological ordering $\tau$ on the gates of $C$; that is, for every gate $x$ of $C$, all the inputs of $x$ are before $x$ in $\tau$.
Construct a vertex ordering $\sigma$ of $G$ as follows:
\begin{enumerate}
\item First, concatenate the vertex orderings $\sigma[A_x]$ for all gates $x$ of $C$ that evaluate to $0$, where the order of concatenation is the reverse of $\tau$.
\item Next, append the vertex orderings $\sigma[B_k],\sigma[B_{k-1}],\ldots,\sigma[B_1]$, in this order.
\item Finally, append the concatentaion of the vertex orderings $\sigma[A_x]$ for all gates $x$ of $C$ that evaluate to $1$, where again the order of concatenation is the reverse of $\tau$.
\end{enumerate}
We now verify that $\sigma$ has degeneracy $2$.
Consider any vertex $u$ of $G$; we check that at most two neighbors of $u$ are placed before $u$ in $\sigma$. 
If $u$ is an internal (i.e. not input or output) vertex of any gadget, then all its neighbors are within the same gadget, and in the vertex ordering of this gadget at most two
neighbors of $u$ were placed before $u$. 
If $u$ is an output vertex of any gadget, then it has four neighbors within this gadget and at most two outside of this gadget. However, all $4$ neighbors within the gadget are placed after $u$ in the vertex
ordering of this gadget, hence $u$ can have at most two neighbors placed before it in $\sigma$.

We are left with the case when $u$ is the input vertex of some gadget.
Consider first the case when $u=\inpv[B_i]$ for some $i\in [k]$.
Since $B_i$ is an $\AND$ gadget, $u$ has two neighbors within $B_i$, both placed before it in $\sigma[B_i]$, and one other neighbor $\outv[B_{i-1}]$. 
By the construction of $\sigma$ and since $y$ evaluates to $1$ in $C$, the vertex $\outv[B_{i-1}]$
is placed after $u$ in $\sigma$, hence $u$ has only two neighbors placed before it in $\sigma$.

Next consider the case when $u=\inpv[A_x]$ for some gate $x$ that evaluates to $0$ in $G$.
If $x$ is a $\FALSE$ gate, then $A_x$ is a copy of the $\AND$ gadget and $u$ has two neighbors within $A_x$, both placed before it in $\sigma$, and one neighbor in some $B_i$, which is placed after it in $\sigma$.
If $x$ is an $\OR$/$\AND$ gate, then $u$ has either one or two neighbors within $A_x$, placed before it in $\sigma$, and one or two neighbors in some other gadgets, say $A_z$ and possibly $A_{z'}$.
However, then $x\geq_\tau z$ and $x\geq_\tau z'$ and hence both $A_z$ and $A_{z'}$ are placed entirely after $A_x$ in $\sigma$ --- regardless whether $z$ and~$z'$ evaluate to $0$ or $1$ in $C$.

Finally, consider the case when $u=\inpv[A_x]$ for some gate $x$ that evaluates to $1$ in $G$.
If $x$ is a $\TRUE$ gate, then $u$ has only two neighbors, both lying within $A_x$ and placed before $u$ in $\sigma$.
If $x$ is an $\OR$ gate, then $u$ has one neighbor within $A_x$, placed before $u$ in $\sigma$, and two neighbors in other gadgets, say $A_z$ and $A_{z'}$, with $x\geq_\tau z$ and $x\geq_\tau z'$. 
Since $x$ evaluates to $1$, either $z$ or $z'$ evaluates to $1$ as well, and consequently the corresponding gadget $A_z$ or $A_{z'}$ is placed entirely after $A_x$ in $\sigma$.
Hence again $u$ has at most two neighbors placed before $u$ in $\sigma$.
If $x$ is an $\AND$ gate, then $u$ has two neighbors within $A_x$, both placed before $u$ in $\sigma$, and one or two neighbors in other gadgets, say $A_z$ and possibly $A_{z'}$, 
with $x\geq_\tau z$ and $x\geq_\tau z'$. 
Since $x$ evaluates to $1$, both~$z$ and $z'$ also need to evaluate to $1$, and hence the corresponding gadgets $A_z$ and $A_{z'}$ are both placed entirely after $A_x$ in $\sigma$.
So again $u$ has at most two neighbors placed before it in $\sigma$.

Having considered all the cases, we conclude that indeed the degeneracy of $\sigma$ is at most $2$.
\end{proof}

\begin{lemma}\label{lem:deg-eval}
If gate $y$ evaluates to $0$ in $C$, then the subgraph of $G$ induced by all gadgets $B_i$, $i\in [k]$, and all gadgets $A_x$ for gates $x$ that evaluate to $0$ in $C$ has minimum degree at least~$3$.
\end{lemma}
\begin{proof}
Let $H$ be this induced subgraph of $G$. Observe that every non-input vertex of every gadget has at least three neighbors already within this gadget.
Since gadgets are included in $H$ in entirety, it follows that all non-input vertices contained in $H$ have degree at least $3$ in $H$.

It remains to show that every input vertex $u$ in $H$ also has degree at least $3$ in $H$.
Suppose first that $u=\inpv[B_i]$ for some $i\in [k]$.
Since $B_i$ is a copy of the $\AND$ gadget, $u$ has two neighbors within $B_i$ and one neighbor being $\outv[B_{i-1}]$.
All gadgets $A_y=B_0,B_1,\ldots,B_k$ are included in $H$, because $y$ evaluates to $0$ in $C$, so in particular $\outv[B_{i-1}]$ is also included in $H$.

Now suppose that $u=\inpv[A_x]$ for some gate $x$ that evaluates to $0$ in $C$.
If $x$ is a $\FALSE$ gate, then $u$ has two neighbors within $A_x$ and one neighbor in some $B_i$, which is also included in $H$.
If $x$ is an $\OR$ gate, then $u$ has one neighbor within $A_x$ and two neighbors in other gadgets $A_z$ and $A_{z'}$, where $z$ and $z'$ are the inputs to $x$.
Since $x$ evaluates to $0$ in $C$, both $z$ and $z'$ have to evaluate to $0$ in $C$ as well, hence both $A_z$ and $A_{z'}$ are included in $H$.
Consequently, $u$ has three neighbors in $H$.
Finally, if $x$ is an $\AND$ gate, then $u$ has two neighbors within $A_x$ and one or two neighbors in other gadgets, say $A_z$ and possibly $A_{z'}$, where $z$ and $z'$ are the inputs to $x$.
Again, since $x$ evaluates to $0$ in $C$, either $z$ or $z'$ have to evaluate to $0$ in $C$ as well, hence at least one of $A_z$ and $A_{z'}$ is included in $H$.
So again $u$ has three neighbors in $H$.

Having considered all the cases, we conclude that $H$ indeed has minimum degree at least $3$.
\end{proof}

\cref{lem:correctness} now directly follows from \cref{lem:eval-deg} and \cref{lem:deg-eval}: \cref{lem:eval-deg} provides the left-to-right implication, 
while \cref{lem:deg-eval} in combination with \cref{prop:subgraph}
provides the right-to-left implication. This concludes the proof of \cref{thm:degeneracy-Phard}.

At the end we would like to remark that in \cref{thm:degeneracy-Phard} the constant $2$ can be replaced by any integer $c\geq 2$. To see this, consider adding $c-2$ universal vertices 
(i.e. adjacent to all other vertices) to the graph $G$ obtained in the reduction. It is not hard to see that this increases the degeneracy of the graph by exactly $c-2$.

%% file: gadgets.pdf_tex
\begingroup%
  \makeatletter%
  \providecommand\color[2][]{%
    \errmessage{(Inkscape) Color is used for the text in Inkscape, but the package 'color.sty' is not loaded}%
    \renewcommand\color[2][]{}%
  }%
  \providecommand\transparent[1]{%
    \errmessage{(Inkscape) Transparency is used (non-zero) for the text in Inkscape, but the package 'transparent.sty' is not loaded}%
    \renewcommand\transparent[1]{}%
  }%
  \providecommand\rotatebox[2]{#2}%
  \ifx\svgwidth\undefined%
    \setlength{\unitlength}{306.2682952bp}%
    \ifx\svgscale\undefined%
      \relax%
    \else%
      \setlength{\unitlength}{\unitlength * \real{\svgscale}}%
    \fi%
  \else%
    \setlength{\unitlength}{\svgwidth}%
  \fi%
  \global\let\svgwidth\undefined%
  \global\let\svgscale\undefined%
  \makeatother%
  \begin{picture}(1,0.61764534)%
    \put(0,0){\includegraphics[width=\unitlength]{gadgets.pdf}}%
    \put(0.24060293,0.03417621){\color[rgb]{0,0,0}\makebox(0,0)[lb]{\smash{$\inpv$}}}%
    \put(0.00451006,0.50713145){\color[rgb]{0,0,0}\makebox(0,0)[lb]{\smash{$\outv$}}}%
    \put(0.14798821,0.60106023){\color[rgb]{0,0,0}\makebox(0,0)[lb]{\smash{$\OR$ gadget}}}%
    \put(0.8199564,0.15497862){\color[rgb]{0,0,0}\makebox(0,0)[lb]{\smash{$\inpv$}}}%
    \put(0.58081029,0.50783996){\color[rgb]{0,0,0}\makebox(0,0)[lb]{\smash{$\outv$}}}%
    \put(0.72428841,0.60176874){\color[rgb]{0,0,0}\makebox(0,0)[lb]{\smash{$\AND$ gadget}}}%
  \end{picture}%
\endgroup%

%% file: app-adm.tex
\section{Proof of \cref{lem:adm-half}}\label{app:adm-half}

Our presentation is based on the proof of \cref{lem:adm-obstacle} in \cite[Lemma~8 of Chapter~2]{notes}, while the proof itself is a generalization of the proof of Grohe et al.~\cite{grohe2015colouring}.

\begin{aloneproof}[of \cref{lem:adm-half}]
For the sake of contradiction suppose $G$ does not admit any depth-$(r-1)$ topological minor with edge density larger than $d$.
Let $\ell\coloneqq 6\varepsilon^{-1}rd^3$.
Let $U\subseteq S$ be the set of those vertices $v\in S$ for which $\backconn_r(S,v)>\ell$.
We know that $|U|\geq \varepsilon|S|$.
For each $v\in U$, let us fix any path family $\Pp_v$ witnessing $\backconn_r(S,v)>\ell$.
That is, $\Pp_v$ consists of more than $\ell$ paths in $G$ such that each $P\in \Pp_v$ has length at most $r$, leads from $v$ to another vertex of $S$, and all its internal vertices are outside of $S$, 
and moreover all paths in $\Pp_v$ pairwise share only the vertex $v$.

Let $\Qq$ be an inclusion-wise maximal family of paths in $G$ satisfying the following conditions:
\begin{itemize}
\item Each path $Q\in \Qq$ has length at most $2r-1$, connects two different vertices in $S$, and all its internal vertices do not belong to $S$.
\item Paths from $\Qq$ are pairwise internally vertex-disjoint (i.e. they can share only the endpoints).
\item For every pair of distinct vertices $u,v\in S$, there is at most one path from $\Qq$ that connects $u$ and $v$.
\end{itemize}
Consider a graph $H$ on the vertex set $S$ where $u,v\in S$ are adjacent if and only if there is a path in $\Qq$ connecting $u$ and $v$.
Observe that paths in $\Qq$ witness that $H$ is a depth-$r$ topological minor of $G$. Therefore, the edge density of $H$ is at most $d$, implying that
$|\Qq|\leq d|S|$.

Observe further that every subgraph of $H$ is also a depth-$(r-1)$ topological minor of $G$, and hence has edge density at most $d$.
By the hand-shaking lemma, a graph of edge density at most $d$ contains a vertex of degree at most $2d$.
Consequently, every subgraph of $H$ has minimum degree at most $2d$, which means, by \cref{prop:subgraph}, that $H$ is $2d$-degenerate.
By \cref{prop:greedy}, $H$ admits a proper coloring $\lambda$ with $(2d+1)$-colors.

Coloring $\lambda$ partitions $U$ into $2d+1$ color classes. Let $I\subseteq U$ be the largest of them; then $|I|\geq \frac{|U|}{2d+1}\geq \frac{\varepsilon|S|}{2d+1}\geq \frac{\varepsilon|S|}{3d}$.
Note that $I$ is an independent set in $H$.

Let $K\coloneqq S\cup \bigcup_{Q\in \Qq} V(Q)$. Since each path $Q\in \Qq$ contains at most $2r-2$ internal vertices lying outside of $S$, and $|\Qq|\leq d|S|$, we have
\begin{equation}\label{eq:K}
|K|\leq (1+(2r-2)d)|S|\leq 2rd|S|.
\end{equation}
For every $v\in I$, we construct a path family $\Pp'_v$ from $\Pp_v$ by trimming paths as follows. 
For a path $P\in \Pp$, let $u$ be the first (closest to $v$) vertex of $P$ that belongs to $K$; such a vertex always exists since the other endpoint of $P$ belongs to $S$.
Then we define $P'$ as the prefix of $P$ from $v$ to $u$, and we let $\Pp'_v$ to consist of all paths $P'$ for $P\in \Pp_v$. The following assertions follow directly from the construction:
\begin{itemize}
\item $|\Pp_v'|=|\Pp_v|$;
\item each path $P'\in \Pp_v'$ has length at most $r$, connects $v$ with another vertex of $K$, and all its internal vertices do not belong to $K$; and
\item paths from $\Pp_v'$ pairwise share only $v$, and in particular their endpoints other than $v$ are pairwise different.
\end{itemize}
Let $\Rr\coloneqq \bigcup_{v\in I} \Pp'_v$.
The following claim is the crucial step in the proof.

\begin{claim}\label{cl:crux}
Any two distinct paths $R,R'\in \Rr$ are internally vertex-disjoint and do not have the same endpoints. 
\end{claim}
\begin{clproof}
Suppose $R\in \Pp'_v$ and $R'\in \Pp'_{v'}$ for some $v,v'\in I$. If $v=v'$ then the claim follows from the properties of $\Pp'_v$ stated above, hence suppose otherwise.
Suppose first that $R$ and $R'$ intersect at some vertex $w\notin K$; in particular $w$ is an internal vertex of both $R$ and $R'$. 
Consider the union of the prefix of $R$ from $v$ to $w$ and the prefix of $R'$ from $v'$ to $w$. 
This union contains a path of length at most $2r-2$ connecting $v$ and $v'$, whose all internal vertices do not belong to $K$; call this path $T$.
Since $v,v'\in I$, $v$ and $v'$ are non-adjacent in $H$, which means that in $\Qq$ there is no path connecting $v$ and $v'$.
It follows that $T$ could be added to $\Qq$ without spoiling any of the conditions imposed on $\Qq$, a contradiction with the maximality of $\Qq$.

We are left with verifying that it is not the case that $R$ and $R'$ have exactly the same endoints $v$ and $v'$.
But in this case $R$ would be a path of length at most $r$ connecting $v$ and $v'$ that would be disjoint from $K$. So again $R$ could be added to $\Qq$ without spoiling any of the conditions
imposed on $\Qq$, which would contradict the maximality of $\Qq$.
\end{clproof}

Consider now a graph $J$ on the vertex set $K$ where two vertices $v,w$ are adjacent if and only if there is a path in $\Rr$ that connects them.
By \cref{cl:crux} the paths in $\Rr$ witness that $J$ is a depth-$(r-1)$ topological minor of $G$, so in particular the edge density of $J$ is at most $d$, implying $|E(J)|\leq d|K|$.
On the other hand, by \cref{cl:crux} every path in $\Rr$ gives rise to a different edge in $E(J)$, implying that $|E(J)|\geq |\Rr|$.
By combining this with \cref{eq:K} we infer that
$$2rd^2|S|\geq d|K|\geq |\Rr|>|I|\cdot \ell\geq \frac{\varepsilon|S|}{3d}\cdot 6\varepsilon^{-1}rd^3=2rd^2|S|.$$
This is a contradiction.
\end{aloneproof}